\documentclass[preprint,3p]{elsarticle}
\newtheorem{theorem}{Theorem}
\newdefinition{remark}[theorem]{Remark}
\newproof{proof}{\textit{Proof}}
\newdefinition{definition}[theorem]{Definition}
\newdefinition{example}[theorem]{Example}

\usepackage{amssymb}
\usepackage{amsmath}
\usepackage{mathrsfs}
\usepackage{eucal}
\usepackage[charter,cal=cmcal]{mathdesign}
\usepackage[utf8]{inputenc}
\usepackage[T1]{fontenc}
\usepackage[colorlinks]{hyperref}
\usepackage{url}
\usepackage{microtype}
\usepackage{soul}

\usepackage[small,bold,langfont=roman]{complexity}
\newclass{\HOCA}{HOCA}
\newcommand{\pnuca}{\ensuremath{\text{p-}\nu\text{CA}}\xspace}
\newcommand{\N}{\ensuremath{\mathbb{N}}\xspace}
\newcommand{\set}[1]{\{#1\}}

\newcommand{\structure}[1]{\ensuremath{\left\langle#1\right\rangle}}
\newcommand{\myvec}[1]{\ensuremath{\boldsymbol{#1}}\xspace}
\renewcommand{\v}{\myvec{v}}

\newcommand{\x}{\myvec{x}}
\newcommand{\ccc}{\myvec{c}}

\usepackage{aurical,bbm}
\newcommand{\pp}{\mathbbm{p}}
\newcommand{\mm}{\mathbbm{m}}
\newcommand{\uu}{\mathbbm{u}}
\renewcommand{\ll}{\mathbbm{l}}

\renewcommand{\M}{\myvec{M}}
\newcommand{\NN}{\myvec{N}}
\newcommand{\LL}{\myvec{L}}
\newcommand{\U}{\myvec{U}}
\renewcommand{\P}{\myvec{P}}

\newcommand{\mat}[2]{Mat\left(#1,#2\right)}

\newcommand{\Z}{\mathbb{Z}}
\newcommand{\LP}{\ensuremath{Z\left[X,X^{-1}\right]}\xspace}
\newcommand{\LPZ}{\ensuremath{\Z\left[X,X^{-1}\right]}\xspace}
\newcommand{\LPm}{\ensuremath{\Z_m\left[X,X^{-1}\right]}\xspace}
\newcommand{\LPmn}{\ensuremath{\Z_m^n\left[X,X^{-1}\right]}\xspace}
\newcommand{\LPK}{\ensuremath{\Z_{p^k}\left[X,X^{-1}\right]}\xspace}

\usepackage{xspace, comment}

\newcommand{\CS}{\ensuremath{\Z^{\Z_m^{k+1}}}\xspace}

\newcommand{\ie}{i.e.\@\xspace}

\newcommand{\modulo}[2]{\left[ #1\right]_{#2}}

\newcommand{\z}{\mathbb{Z}}
\newcommand{\n}{\mathbb{N}}
\newcommand{\az}{S^{\mathbb{Z}}}

\newtheorem{lemma}[theorem]{Lemma}

\newtheorem{proposition}[theorem]{Proposition}

\usepackage{xcolor}

\usepackage{pst-node}

\usepackage{tikz-cd} 
\usetikzlibrary{arrows}

\newtheorem{teo}{Theorem}[section]

\newtheorem{notation}[teo]{Notation}

\def \� {\`e }
\def \� {\`a }
\def \� {\`i }
\def \� {\`u }
\def \� {\`o }

\newcounter{exnum}
\newenvironment{ex}{\refstepcounter{exnum}%
                         \begin{trivlist}
                           \item[\hskip\labelsep{\bf Example
                                                     \arabic{exnum}}]}%
                        {\hfill\hbox{$\quad\Box$}\end{trivlist}}

\begin{document}

\title{On the dynamical behaviour of linear higher-order \\
cellular automata and its decidability}
%\title{Decidability of Sensitivity to the initial conditions and Equicontinuity \\for Linear High-Order Cellular Automata}
%\tnoteref{blabla}}
%\tnotetext[blabla]{This paper collects, in a more extended/generalized version, both results presented at CiE 2014 conference~\cite{Formenti2014a} and the ones presented at DCFS 2014 conference~\cite{Formenti2014b}. It also contains a considerable number of further results.}
\author[unimib]{Alberto Dennunzio}
\ead{dennunzio@disco.unimib.it}

\author[unice]{Enrico Formenti}
\ead{enrico.formenti@unice.fr}

\author[unimib]{Luca Manzoni}
\ead{luca.manzoni@disco.unimib.it}

\author[unibo]{Luciano Margara}
\ead{luciano.margara@unibo.it}

\author[unimib,amu]{Antonio E. Porreca}
\ead{antonio.porreca@lis-lab.fr}

\address[unimib]{Dipartimento di Informatica, Sistemistica e Comunicazione,
  Università degli Studi di Milano-Bicocca,
  Viale Sarca 336/14, 20126 Milano, Italy}
  
\address[unice]{Universit\'e C\^ote d'Azur, CNRS, I3S, France}

\address[unibo]{Department of Computer Science and Engineering, University of Bologna, Cesena Campus, Via Sacchi 3, Cesena, Italy}

\address[amu]{Aix Marseille Université, Université de Toulon, CNRS, LIS, Marseille, France}

\begin{abstract}
\begin{comment}The paper starts the exploration of the dynamical behaviuor of high-order cellular automata (HOCA). In particular, decidable conditions for the sensitivity to initial conditions of linear HOCA are provided. We also prove an equivalence between
linear HOCA and linear non-uniform cellular automata which further motivate
the approach since little is known for linear non-uniform cellular automata.
\end{comment}
Higher-order cellular automata (HOCA) are a variant of cellular automata (CA) used in many applications (ranging, for instance, from the design of secret sharing schemes to data compression and image processing), and in which the global state of the system at time $t$ depends not only on the state at time $t-1$, as in the original model, but also on the states at time $t-2$, \ldots, $t-n$, where $n$ is the memory size of the HOCA. We provide decidable characterizations of two important dynamical properties, namely, sensitivity to the initial conditions and equicontinuity, for linear HOCA over the alphabet $\Z_m$.
Such characterizations extend the ones shown in~\cite{ManziniM99} for linear CA (LCA) over the alphabet $\Z^{n}_m$ in the case $n=1$.  We also prove that linear HOCA of size memory $n$ over $\Z_m$ form a class that is indistinguishable from a specific subclass of LCA over $\Z_m^n$. This enables to decide injectivity and surjectivity for linear HOCA of size memory $n$ over $\Z_m$ using the decidable characterization provided in~\cite{bruyn1991} and~\cite{kari2000} for injectivity and surjectivity of LCA over $\Z^n_m$.  Finally,  we prove an equivalence between
LCA over $\Z_m^n$ and an important class of non-uniform CA, another variant of CA  used in many applications. % as, for instance, compression of images and signals.
%which further motivate
%the approach since little is known for linear non-uniform cellular automata.
%linear Cellular Automata over $\Z_m^n$, used in many applications. 
\end{abstract}


\maketitle

%%%INTRODUZIONE
\section{Introduction}
Cellular automata (CA) is well-known formal model which has been successfully applied in
a wide number of fields to simulate complex phenomena involving local, uniform, and synchronous processing (for recent results and an up-to date bibliography on CA, see~\cite{DennunzioFW14, DennunzioGM08, AcerbiDF07, DennunzioLFM13}). 
More formally, a CA is made of an infinite set of identical finite automata arranged over a regular cell grid (usually $\Z^d$ in dimension $d$) and all taking a state from a finite set $S$ called the \emph{set of states} or the \emph{alphabet} of the CA. In this paper, we consider one-dimensional CA. A \emph{configuration} is a snapshot  of all states of the automata, \ie, a function $c:\Z\to S$. A \emph{local rule} updates the state of each automaton on the basis of its current state and the ones of a finite set of neighboring automata. All automata are updated synchronously. 
In the one-dimensional settings, a CA over (the alphabet) $S$ is a structure $\structure{S, r, f}$ where $r\in\N$ is the \emph{radius} and  $f:S^{2r+1}\to S$ is the local rule which updates, for each $i\in\Z$, the state of the automaton in the position $i$ of the grid $\Z$ on the basis of states of the automata in the positions $i-r, \ldots, i+r$. A configuration is an element of $S^{\Z}$ and describes the (global) state of the CA. The feature of synchronous updating induces the following \emph{global rule} $F:\az\to\az$ defined as 
\[
\forall c\in \az, \forall i\in\Z,\qquad F(c)_i=f(c_{i-r}, \ldots c_{i+r})\enspace.
\]
As such, the global map $F$ describes the change from any configuration $c$ at any time $t\in\N$ to the configuration $F(c)$ at $t+1$ and summarises the main features of the CA model, namely, the fact that  it is defined through a local rule which is applied uniformly and synchronously to all cells.

Because of a possible inadequacy, in some contexts, of every single one of the three defining features, 
%uniformity and global synchronous clocking. 
%However, these three last
%properties are not suitable in all contexts. Therefore, 
variants of the original CA 
model started appearing, each one relaxing one among these three features.
Asynchronous CA relax synchrony (see~\cite{INGERSON198459,SCHONFISCH1999123, DennunzioFMM13} for instance), non-uniform
CA  relax uniformity~(\cite{DENNUNZIO201473,DENNUNZIO201232}), while hormone CA (for instance)
relax locality~\cite{cervelle:hal-01615278}. However, from the mathematical point of view all those systems, as well
as the original model, fall in the same class, namely, the class of autonomous discrete dynamical systems (DDS) 
and one could also precise \emph{memoryless} systems. Indeed, the latter compute their
next global state just on the basis of their current state, while the past ones play no active role.
Allowing the original model to take into account past states leads to a further natural variant which can further extend the application range of the model itself.

As a motivating example, consider the following classical network routing problem. Assume to have two packet sources
$A$ and $B$ which are connected to a common router $R$ which in
its turn is connected to two receiving hosts $O_1$ and $O_2$ as
illustrated below.
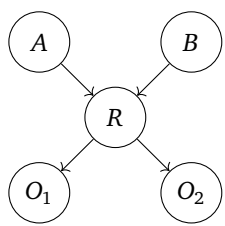
\begin{center}
\begin{tikzpicture}[
mynode/.style={draw,circle,minimum width=.8cm}
]
\node[mynode](A) at (0,0) {$A$};
\node[mynode](B) at (2,0) {$B$};
\node[mynode](R) at (1,-1) {$R$};
\node[mynode](O1) at (0,-2) {$O_1$};
\node[mynode](O2) at (2,-2) {$O_2$};
\draw[->] (A) -- (R);
\draw[->] (B) -- (R);
\draw[->] (R) -- (O1);
\draw[->] (R) -- (O2);
\end{tikzpicture}
\end{center}
%\end{figure}
If $R$ receives a packet $m_A$ from $A$ but none from $B$, then it sends $m_A$ to the output hosts $O_1$ and $O_2$; if both $A$ and $B$ send a packet, $m_A$ and $m_B$,
respectively, then $m_B$ is enqueued and $m_A$ is transmitted
to $O_1$ and $O_2$. Of course, the more frequently simultaneous packets from $A$ and $B$ arrive at $R$, the longer the queue
has to be in order to avoid packet loss. When a whole network is considered, this routing problem can
be easily solved using a variant of CA in which the state of each node keeps
track of the current state of the router and the past
states which represent the received but not yet transmitted
packets.

As to the possible variants of the original CA model, in
\begin{comment}~\cite{KUTRIB1999239}, Kutrib proposed for instance to replace the finite state automaton which is put at each site of the
lattice with a pushdown automaton. The stack of the pushdown automaton allows the memorisation of past states of the current automaton and
of the ones of its neighbours. In
\end{comment}
~\cite{toffoli1977}, Toffoli introduced \emph{higher-order CA}  (HOCA), \ie, variants of CA in which the updating of the state of a cell also depends on the past states of the cell itself and its neighbours. In particular, he showed that any arbitrary reversible \emph{linear} HOCA can be embedded in a  reversible \emph{linear CA} (LCA), where linear means that the local rule is linear. Essentially, the trick consisted in memorizing past states and recover them later on. Some years later, Le Bruyn and Van Den Bergh explained and
generalized the Toffoli's construction and proved that any linear
HOCA having the ring $S=\Z_m$ as alphabet and memory size $n$ can be simulated by a LCA over the alphabet $\Z^n_m$ (see the precise definition in Section~\ref{sec:def})~\cite{bruyn1991}.
In this way a practical way to decide injectivity (which is equivalent to reversibility in this setting) and, as we will see in Section~\ref{sec:def},  surjectivity of HOCA can be easily derived by the characterization of the these properties for the corresponding LCA simulating them. Indeed, in~\cite{bruyn1991} and~\cite{kari2000}, characterizations of injectivity and surjectivity of a LCA over $\Z^n_m$ are provided in terms of properties of the determinant of the matrix associated with it, where the determinant turns out to be an other LCA (over $\Z_m$ and then a LCA simpler than that over $\Z^n_m$). Since the properties of LCA over $\Z_m$ (\ie, LCA over $\Z^n_m$ with $n=1$) have been extensively studied and related decidable characterizations  have been obtained~\cite{ManziniM99a, CattaneoDM04, DamicoMM03, CattaneoFMM00}, one derives the algorithms to decide injectivity and surjectivity for LCA over $\Z^n_m$ and, then, as we will see in Section~\ref{sec:def}, also for HOCA over $\Z_m$ of memory size $n$, by means of the associated matrix.    
%that a LCA
%is injective (resp., surjective) iff the associated matrix is invertible
%(resp., the determinant of the associated matrix is a monomial). 
\begin{comment}
Moreover, Kari provided decomposition results of linear reversible
\HOCA\ into elementary (in some specific sense not to be confused with the usual term of elementary cellular automata) linear
\HOCA. In this way, one can build non-trivial reversible linear \HOCA\ without testing determinants.
\end{comment}

Applications of HOCA (in particular the linear ones) cover a wide span of topics, ranging from
the design of secret sharing schemes~\cite{MARTINDELREY20051356,chai2005}
to data compression and image processing~\cite{Gu2000}. %Finally, remark that in the last references, 
Remark that (linear) HOCA are often required to exhibit a chaotic behavior in order they can be used in applications, as for instance in those above mentioned. The purpose of the present paper is to study, in the context of linear HOCA,  sensitivity to the initial conditions and equicontinuity, where the former is the well-known basic component and essence of the chaotic behavior of a DDS, while the latter represents a strong form of stability. To do that, we put in evidence that any linear HOCA of memory size $n$ over $\Z_m$  is not only simulated by, but also topologically conjugated to a LCA over $\Z^n_m$ defined by a matrix having a specific form. Thus, in order to decide injectivity and surjectivity for linear HOCA  of memory size $n$ over $\Z_m$, by means of that specific matrix one can use the decidable characterization provided in~\cite{bruyn1991} and~\cite{kari2000} for deciding the same properties for LCA over $\Z^n_m$ . As main result, we prove that sensitivity to the initial conditions and equicontinuity are decidable properties for linear HOCA of memory size $n$ over $\Z_m$(Theorem~\ref{main}). In particular we provide a decidable characterization of those properties, in terms of the matrix associated with a linear HOCA. Remark that if $n=1$, starting from our characterizations, one recover exactly the well known ones of sensitivity and equicontinuity for LCA over $\Z_m$. Finally,  we prove an equivalence between
LCA over $\Z_m^n$ and an important class of linear non-uniform cellular automata. This result gives strong motivations to further study LCA over $\Z_m^n$ in the next future. First of all, non-uniform cellular automata is indeed another variant of cellular automata which is used in many applications (in particular, the linear ones). For instance, as pointed out in~\cite{kari2000}, linear non-uniform cellular automata can be used as subband encoders for compressing signals and images~\cite{Shapiro93}. Moreover, little is known for linear non-uniform cellular automata. %
%
%
%
%correlate the properties
%of the matrix characterizing a linear \HOCA\ with sensitivity to initial conditions, a property which is a %basic ingredient of many definitions of chaotic behavior. 
%
\section{Higher-Order CA and Linear CA}
\label{sec:def}
We begin by reviewing some general notions and introducing notations we will use throughout %rest of 
the paper.
\smallskip

\noindent
A \emph{discrete dynamical system} (DDS) is a pair $(\mathcal{X}, \mathcal{F})$ where %$(\mathcal{X}, \mathcal{d})$ 
$\mathcal{X}$ is a space equipped with a metric, \ie, a metric space, and $\mathcal{F}$ is a transformation on $\mathcal{X}$ which is continuous with respect to that metric. 
%, \ie, as pairs consisting of a metric space (with the metric defined above) and a continuous map.\\
The \emph{dynamical evolution} of a DDS $(\mathcal{X}, \mathcal{F})$  starting from the initial state $x^{(0)}\in \mathcal{X}$ is the sequence $\{x^{(t)}\}_{t\in\N}\subseteq \mathcal{X}$ where $x^{(t)}=\mathcal{F}^t(x^{(0)})$ for any $t\in\N$. 

When $\mathcal{X}=\az$ for some set finite $S$, $\mathcal{X}$ is usually equipped with the metric $d$ defined as follows
\[
\forall c,c'\in\az, \quad d(c,c')=\frac{1}{2^n}\;\quad \text{where}\;n=\min\{i\geq 0\,:\,c_i\ne
c'_i \;\text{or}\;c'_{-i}\ne c'_{-i}\}\enspace.
\]
Recall that  $\az$ is a compact, totally
disconnected and perfect topological space (\ie, $\az$ is a Cantor
space).

Any CA $\structure{S, r, f}$ defines the DDS $(\az, F)$, where $F$ is the CA global rule (which is continuous by Hedlund's Theorem~\cite{hedlund69}). From now on, for the sake of simplicity, we will sometimes identify a CA with its global rule $F$ or with the DDS $(\az, F)$. 

Recall that two DDS $(\mathcal{X},\mathcal{F})$ and $(\mathcal{X}',\mathcal{F}')$ are %\emph{isomorphic} (resp., 
\emph{topologically conjugated} if
there exists a homeomorphism $\phi: \mathcal{X} \mapsto \mathcal{X}'$ such that $\mathcal{F}'\circ \phi=\phi\circ \mathcal{F}$, while the \emph{product} of $(\mathcal{X}, \mathcal{F})$ and $(\mathcal{X}', \mathcal{F}')$ is the DDS $(\mathcal{X}\times \mathcal{X}', \mathcal{F}\times \mathcal{F}')$ where $\mathcal{F}\times \mathcal{F}'$ is defined as $\forall (x,x')\in \mathcal{X} \times \mathcal{X}'$, $(\mathcal{F}\times \mathcal{F}')(x,x')=(\mathcal{F}(x), \mathcal{F}'(x'))$ and the space $\mathcal{X}\times \mathcal{X}'$ is as usual endowed with the infinite distance. 
%$d_{\infty}$ such that $d_{\infty}(x,y)(x',y'))=\max\{d(x,x'),d(y,y')\}$ for every pair $(x,y),(x',y')\in X\times Y$.

%\renewcommand{\LP}{\ensuremath{\Z_{p^k}\left[X,X^{-1}\right]}\xspace}
%\paragraph{Notations}
\begin{notation}
For all $i,j\in\z$ with $i\leq j$, we write $[i,j]=\{i,i+1,\ldots,j\}$ %(resp., $[i,j)=\{i,i+1,\ldots,j-1\}$) 
to denote the interval of integers between $i$ and $j$.  For any $n\in\N$ and any set $Z$  the set of all $n\times n$ matrices with coefficients in $Z$ and the set of Laurent polynomials with coefficients in $Z$ will be noted by $\mat{n}{Z}$ and $\LP$, respectively. In the sequel, bold symbols are used to denote vectors, matrices, and configurations over a set of states which is a vectorial space. Moreover, 
$m$ will be an integer bigger than $1$ and $\Z_m=\{0, 1, \ldots, m-1\}$  the ring with the usual sum and product modulo $m$. For any $\x\in\Z^n$ (resp., any matrix $\M(X)\in\mat{n}{\LPZ}$), we will denote by $\modulo{\x}{m}\in\Z_m^n$ (resp., $\modulo{\M(X)}{m}$),  the vector (resp., the matrix) in which each component $x^i$ of $\x$ (resp., every coefficient of each element of $\M(X)$) is taken modulo $m$. Finally, for any matrix $\M(X)\in\LPm$ and any $t\in\N$, the $t$-th power of $\M(X)$ will be noted more simply by $\M^t(X)$ instead of $(\M(X))^t$.
\end{notation}
\begin{definition}[Higher-Order Cellular Automata]
\begin{comment}
A \emph{$k$-th order CA (HOCA)} is a structure
$\mathcal{H}=\structure{k,S,r,h}$ where $k\in\N$ is the \emph{memory} size, $S$ is the \emph{alphabet}, $r\in\N$ is the \emph{radius}, and $f\colon S^{(2r+1)k}\to S$ is the \emph{local rule}. The \emph{dynamical evolution} of an HOCA $\mathcal{H}$  starting from the initial configurations $c^{(0)},c^{(1)},\ldots,c^{(k-1)}\in\az$ is the sequence $\{c^{(t)}\}_{t\in\N}\subset \az$ defined for any $t\geq k$ as follows 
%
%\emph{global rule} $H:\az\to\az$ defined as
\[
\quad \forall i\in\Z,\quad  c_i^{(t)}=h\left(\quad \begin{matrix} c^{(t-k)}_{[i-r, i+r]}\\ \\ c^{(t-k+1)}_{[i-r, i+r]}\\ \vdots\\ \\c^{(t-1)}_{[i-r, i+r]}\end{matrix}\quad \right)
\]
OPPURE (SE SI SCRIVONO LE PROPRIETA' DINAMICHE PER GLI HOCA)
\\
\end{comment}
A %\emph{$k$-th order CA 
\emph{Higher-Order Cellular Automata (HOCA)} is a structure
$\mathcal{H}=\structure{k,S,r,h}$ where $k\in\N$ with $k\geq 1$ is the \emph{memory} size, $S$ is the \emph{alphabet}, $r\in\N$ is the \emph{radius}, and $h\colon S^{(2r+1)k}\to S$ is the \emph{local rule}. Any HOCA $\mathcal{H}$ induces the \emph{global rule} $H:\left({\az}\right)^k\to\left({\az}\right)^k$ associating any vector $\myvec{e}=(e^1, \ldots, e^k)\in \left({\az}\right)^k$ of $k$ configurations of $\az$ with the vector $H(\myvec{e})\in\left({\az}\right)^k$ such that $H(\myvec{e})^j=e^{j+1}$ for each $j\neq k$ and
\[
\quad \forall i\in\Z,\quad H(\myvec{e})^k_i=h\left(\begin{matrix} e^{1}_{[i-r, i+r]}\\ \\ e^{2}_{[i-r, i+r]}\\ \vdots\\ e^{k}_{[i-r, i+r]}\end{matrix}\right)
\]
In this way, $\mathcal{H}$ defines the DDS $\left(\left({\az}\right)^k, H\right )$. As for CA, we sometimes identify a HOCA with its global rule or the DDS defined by it. Moreover, we will often refer to a HOCA over $S$ to stress the alphabet over which the HOCA is defined. 
\begin{comment}
 the \emph{dynamical evolution} of %an HOCA 
$\mathcal{H}$  starting from the initial configuration vector $\ccc^{(0)}\in \left({\az}\right)^k $%=(c^{(0)},c^{(1)},\ldots,c^{(k-1)})\in \left({\az}\right)^k$ 
is the sequence $\{\ccc^{(t)}\}_{t\in\N}\subset \left({\az}\right)^k$ where $\ccc^{(t)}=H^t(\ccc^{(0)})$ for any $t\in\N$.
\end{comment}
\end{definition}
\begin{remark}
It is easy to check that for any HOCA $\mathcal{H}=\structure{k,S,r,h}$ there exists a CA $\structure{S^k,r,f}$ which is topologically conjugated to $\mathcal{H}$.
\end{remark}

The study of the dynamical behaviour of HOCA is still at its early stages; a few results are known for the class of \emph{linear HOCA}, namely, those HOCA defined by a local rule $f$  which is \emph{linear}, \ie,  $S$ is %the ring 
$\Z_m$ %=\{0, 1, \ldots, m-1\}$
%  with the usual sum and product modulo $m$, 
and there exist coefficients $a^j_i\in\Z_m$  ($j=1, \ldots, k$ and $i=-r, \ldots, r$) such that for any element
\[
\x=\left(\begin{matrix} x^{1}_{-r} & \hdots &  x^{1}_{r} \\ \\ x^{2}_{-r} & \hdots &  x^{2}_{r}\\ \vdots\\ x^{k}_{-r} & \hdots &  x^{k}_{r}\end{matrix}\right)\in \Z_m^{(2r+1)k}, \quad f(\x)=\modulo{\sum_{j=1}^{k}\sum_{i=-r}^{r} a^j_i x^{j}_{i}}{m}\enspace .
\]  
%where $\modulo{x}{m}$ is the integer $x$ taken modulo $m$. 
It is easy to see that linear HOCA are additive, \ie, 
\[
\forall\myvec{c},\myvec{d}\in \left(\Z_m^{\Z}\right)^k, \quad  H(\myvec{c}+\myvec{d})=
H(\myvec{c})+H(\myvec{d})
\]
where, with the usual abuse of notation, $+$ denotes the natural extension of the sum over $\Z_m$ to both $\Z_m^{\Z}$ and $\left(\Z_m^{\Z}\right)^k$.

In~\cite{bruyn1991}, a much more convenient representation is introduced for the case of linear HOCA (in dimension $d=1$) by means of the following notion.
%\noindent\paragraph{Standing assumption} $S$ is assumed to be the additive
%group $\Z_m$ for $m\in\N$, the global rule is additive and the dimension $d$ is $1$.
%From now on, we share the same assumptions.

\begin{definition}[Linear Cellular Automata]
A \emph{Linear Cellular Automaton} (LCA) over the alphabet $\Z_m^n$ is a CA $\mathcal{L}=\structure{\Z_m^n, r, f}$ where the local rule $f\colon(\Z_m^n)^{2r+1}\to\Z_m^n$ is defined by $2r+1$ matrixes  
$\M_{-r},\ldots, \M_0,\ldots, \M_r\in\mat{n}{\Z_m}$ 
%$A_{-r},\ldots, A_0,\ldots, A_r\in\mat{n}{\Z_m}$ 
as follows: 
$f(\x_{-r}, \ldots,\x_0,\ldots,\x_r) = 
\modulo{\sum_{i=-r}^r \M_i\cdot\x_i}{m}$ for any 
$(\x_{-r}, \ldots,\x_0,\ldots,\x_r)\in(\Z_m^n)^{2r+1}$.
\end{definition}
\begin{remark}
LCA over $\Z_m^n$ have been strongly investigated in the case $n=1$ and all the dynamical properties have been characterized in terms of the $1\times 1$ matrices (\ie, coefficients) defining the local rule, in any dimension too~\cite{ManziniM99,CattaneoDM04}.  
\end{remark}

We recall that any linear HOCA $\mathcal{H}$ can be simulated by a suitable LCA, as shown in~\cite{bruyn1991}. Precisely, given a linear HOCA $\mathcal{H}=\structure{k,\Z_m,r,h}$, where $h$ is defined by the coefficients $a^j_i\in\Z_m$, the LCA simulating $\mathcal{H}$ is $\mathcal{L}=\structure{\Z_m^k, r, f}$  with $f$ defined by following matrices 
%\[
\begin{align}
\label{M0}
\M_0=
\begin{bmatrix}
0&1&0& \dots&0&0\\
0&0& 1&\ddots&0&0\\
0&0&0 &\ddots&0&0\\
\vdots&\vdots& \vdots&\ddots&\ddots&\vdots\\
%0&0&0&0&1&0&0\\
%0&0&0&0&0&1&0\\
0        &   0 &0   & \dots &0&1\\
a_0^1&a_0^2& a_0^3 &\dots &a_0^{k-1}&a_0^k
\end{bmatrix}\enspace,
%\]
\end{align}
and, for $i\in [-r, r]$ with $i\ne0$,
%\[
\begin{align}
\label{Mi}
\M_i=
\begin{bmatrix}
0&0&0&\dots&0&0\\
0&0&0&\dots&0&0\\
0&0&0&\dots&0&0\\
\vdots&\vdots&\vdots& \ddots&\vdots&\vdots\\
%0&0&0&0&1&0&0\\
%0&0&0&0&0&1&0\\
0        &   0    & 0 &  \dots &0&0\\
a_i^1&a_i^2& a_i^3& \dots &a_i^{k-1}&a_i^k
\end{bmatrix}\enspace.
\end{align}
%\]
\begin{remark}
\label{rem:classes}
We want to put in evidence that a stronger result actually holds (easy proof, important remark): any linear HOCA $\mathcal{H}$ is topologically conjugated to the LCA $\mathcal{L}$ defined by the matrices in~\eqref{M0} and~\eqref{Mi}. %$\M_{-r},\ldots, \M_0,\ldots, \M_r$ above introduced. 
Clearly, the converse also holds: for any LCA %$\mathcal{L}$  
defined by the matrices in~\eqref{M0} and~\eqref{Mi} %$\M_{-r},\ldots, \M_0,\ldots, \M_r$ 
there exists a linear HOCA which is topologically conjugated to it. In other words, up to a homeomorphism the whole class of linear HOCA is identical to the subclass of LCA defined by the matrices above introduced.
% $\structure{\Z_m^n, r, f}$, 
In the sequel, we will call $\mathcal{L}$ the \emph{matrix presentation} of $\mathcal{H}$.  
%, which is  defined %by the matrixes $\M_{-r},\ldots, \M_0,\ldots, \M_r$ above introduced. 
\end{remark}

We are now going to show a stronger and useful new fact, namely, that the class of linear HOCA is nothing but the subclass of LCA represented by a formal power series which is a matrix in Frobenius normal form. Before proceeding, let us recall the  \emph{formal power series} (fps)
\noindent
which have been successfully used to study the dynamical 
behaviour of LCA in the case $n=1$~\cite{ito83,ManziniM99,ManziniM99a,CattaneoFMM00}. The idea of this formalism is that configurations and global rules are represented by
suitable polynomials and the application of the global rule turns into multiplications of
polynomials. 
In the more general case of LCA over $\Z_m^n$,  a configuration $\ccc\in(\Z_m^n)^\Z$ can be associated with the fps
\[
\P_{\ccc}(X)=\sum_{i\in\z}\ccc_i X^i = 
\begin{bmatrix}
c^1(X)\\
\vdots\\
c^n(X)
\end{bmatrix}
=
\begin{bmatrix}
\sum_{i\in\Z} c_i^1 X^i\\
\vdots\\
\sum_{i\in\Z} c_i^n X^i
\end{bmatrix}\in\left(\LPm\right)^n\cong \LPmn \enspace.
\]
\noindent
Then, if 
%$\forall \v\in\az,\forall i\in \Z, \quad F(\v)_i = \modulo{\M_{-r}\v_{i-r}+\cdots+\M_0 \v_i+\cdots +\M_{r}\v_{i+r}}{m}$ 
$F$ is the global rule of a LCA defined by $\M_{-r},\ldots, \M_0,\ldots, \M_r$, one finds 
\[
\P_{F(\ccc)}(X)=\modulo{\M(X) \P_{\ccc}(X)}{m}
\]
where 
$$
\M(X) =\modulo{ \sum\limits_{i=-r}^{r} \M_i X^{-i}}{m}
$$
is the \emph{finite fps}, or, \emph{the matrix}, \emph{associated with the LCA $F$}. In this way, for any integer $t>0$ the fps associated with $F^t$ is $\M(X)^t$, and then 
$
\P_{F^t(\ccc)}(X)=\modulo{\M(X)^t \P_{\ccc}(X)}{m}\enspace.
$ 
Roughly speaking, the action of a LCA over a configuration is given by multiplication between elements of $\mat{n}{\LPm}$ with elements of $\left(\LPm\right)^n$. Throughout this paper, $\M(X)^t$ will refer to $\modulo{\M(X)^t}{m}$. 
%\textbf{Therefore, roughly speaking, the action of a matrix presentation over a configuration is given by multiplication in $\Z_m[X]$}.

%Denote $\LP$ the set of Laurent polynomials with coefficients over $\Z_{m}$. 
%We are interested in the degrees of such polynomials both \wrt $X$ and $X^{-1}$.

A matrix $\M(X)\in\mat{n}{\LP}$ is in \emph{Frobenius normal form} if  
\begin{eqnarray}
\label{FNF}
\M(X)=
\begin{comment}
\begin{bmatrix}
0&1&0& \dots&0&0\\
0&0& 1&\ddots&0&0\\
0&0&0 &\ddots&0&0\\
\vdots&\vdots& \vdots&\ddots&\ddots&\vdots\\
%0&0&0&0&1&0&0\\
%0&0&0&0&0&1&0\\
0        &   0 &0   & \dots &0&1\\
\pp_1(X)&\pp_2(X)& \pp_3(X) &\dots &\pp_{n-1}(X)&\pp_{n}(X)
\end{bmatrix}\\
\end{comment}
\begin{bmatrix}
0&1&0& \dots&0&0\\
0&0& 1&\ddots&0&0\\
0&0&0 &\ddots&0&0\\
\vdots&\vdots& \vdots&\ddots&\ddots&\vdots\\
\\
%0&0&0&0&1&0&0\\
%0&0&0&0&0&1&0\\
0        &   0 &0   & \dots &0&1\\
%\medskip
\\
\mm_0(X)&\mm_1(X)& \mm_2(X) &\dots &\mm_{n-2}(X)&\mm_{n-1}(X)
\end{bmatrix}
\end{eqnarray}
\begin{comment}
\[
M(X)=
\begin{bmatrix}
0&1&0& \dots&0&0\\
0&0& 1&\ddots&0&0\\
0&0&0 &\ddots&0&0\\
\vdots&\vdots& \vdots&\ddots&\ddots&\vdots\\
%0&0&0&0&1&0&0\\
%0&0&0&0&0&1&0\\
0        &   0 &0   & \dots &0&1\\
c_0(X)&c_1(X)& c_2(X) &\dots &c_{n-2}(X)&c_{n-1}(X)
\end{bmatrix}
\]
\end{comment}
\begin{comment}
\[
M(X) =\begin{bmatrix}
 0& \cdots & \cdots   &  0 &   | \\
1  & \ddots &  &\vdots  &  |\\
0  &\ddots  &\ddots  &\vdots  & \myvec{c}(X) \\
\vdots   &  & \ddots & 0 &  |\\
 0   & \cdots & 0 & 1 &  | \\
 \end{bmatrix} 
\]
\end{comment}
%for some vector 
%$\myvec{c}(X)=\begin{bmatrix} c_0(X)\\ c_1(X)\\ \vdots\\ c_{n-1}(X)\end{bmatrix}$
%%% PRECEDENTEwhere each $\pp_i(X)\in\LP$. %ciao $\pp^M(X)$
where each $\mm_i(X)\in\LP$

From now on, for a  given matrix $\M(X)\in\mat{n}{\LP}$ in Frobenius normal form, $\myvec{\mm}(X)$ 
%$\myvec{\pp}(X)$
will always make reference to its $n$-th row. %column vector. 

\begin{definition}[Frobenius LCA]
  \label{frobeniusLCA}
%Let $p$ be a prime number and $k\geq 1$. 
A LCA $F$ over the alphabet $\Z^n_{m}$ is said to be a \emph{Frobenius LCA} if the fps $\M(X)\in\mat{n}{\LPm}$ associated with $F$ is in Frobenius normal form. 
\end{definition}
\begin{comment}
$$
\begin{bmatrix}
 0& \cdots & \cdots   &  0 &c_0(X)    \\
1  & \ddots &  &\vdots  &c_1(X)  \\
0  &\ddots  &\ddots  &\vdots  &\vdots  \\
\vdots   &  & \ddots & 0 & \vdots \\
 0   & \cdots & 0 & 1 & c_{n-1}(X)  \\
 \end{bmatrix} 
$$
where each $c_i(X)$ is a Laurent polynomial with coefficients belonging to $\Z_{p^k}$.
\end{comment}
\begin{comment}
Sia $F$ un MCA di frobenius definito come segue.
$$F =\begin{bmatrix}
 0& \cdots & \cdots   &  0 &c_0    \\
1  & \ddots &  &\vdots  &c_1  \\
0  &\ddots  &\ddots  &\vdots  &\vdots  \\
\vdots   &  & \ddots & 0 & \vdots \\
 0   & \cdots & 0 & 1 & c_{n-1}  \\
 \end{bmatrix} 
$$
con $c_i$ polinomi di Laurent con coefficienti in $\Z_{p^k}$ con $k\geq 1$ e $p$ primo.
\end{comment}
%%
%\begin{remark}
\begin{comment}
Since for any given configuration $\ccc\in (\Z_m^n)^{\Z}$
\[
P_{\ccc}(X) = 
\begin{bmatrix}
c^1(X)\\
\vdots\\
c^n(X)
\end{bmatrix}
=
\begin{bmatrix}
\sum_{i\in\Z} c_i^1 X^i\\
\vdots\\
\sum_{i\in\Z} c_i^n X^i
\end{bmatrix}
\]
\end{comment}
\begin{comment}
Remark that the action of a Frobenius LCA $F$ over any given configuration $\ccc\in (\Z_m^n)^{\Z}$ is such that
\[
\P_{F(\ccc)}(X) = \M(X) \P_{\ccc}(X) =
\begin{bmatrix}
c^2(X)\\
\vdots\\
\sum_{i=1}^n \pp_i(X) c^i(X)
\end{bmatrix} \enspace.
\]
\end{comment}
%PARLARE DI CONIUGAZIONE TOPOLOGICA
%This allows us to state the following
It is immediate to see that a LCA is a Frobenius one iff it is defined by the matrices in~\eqref{M0} and~\eqref{Mi}, \ie, iff it is topologically conjugated to a linear HOCA. This fact together with Remark~\ref{rem:classes} and Definition~\ref{frobeniusLCA}, allow us to state the following
\begin{proposition} 
\label{prop:homeo}
Up to a homeomorphism, the class of linear HOCA over $\Z_m$ of memory size $n$ is nothing but the class of Frobenius LCA over $\Z^n_m$.
\end{proposition}

At this point we want to stress that the action of a Frobenius LCA $F$, or, equivalently, a linear HOCA, over a configuration $\ccc\in (\Z_m^n)^{\Z}$ can be viewed as a bi-infinite array of \emph{linear--feedback shift register}. %~\cite{}. 
Indeed, it holds that
\[
\P_{F(\ccc)}(X) = \M(X) \P_{\ccc}(X) =
\begin{bmatrix}
c^2(X)\\
\vdots\\
\sum_{i=1}^n \mm_{i-1}(X) c^i(X)
\end{bmatrix} \enspace,
\]
where the $n$--th component of the vector $\P_{F(\ccc)}(X)$ is given by the sum of the results of the actions of $n$ one-dimensional LCA over $\Z_m$ each of them applied on a different component of $\P_{\ccc}(X)$, or, in other words, on a different element of the memory of the linear HOCA which is topologically conjugated to $F$.

\begin{comment}
As a consequence, given a Frobenius LCA $F$, for any configuration \ccc\in (\Z_m^n)^{\Z}$

 the $n$th component of the vector $\P_{F(\ccc)}(X)$ is given by the sum of the results of the actions of $n$ one-dimensional LCA over $\Z_m$ each of them applied on a different component of $\P_{\ccc}(X)$, or, in other words, on a different element of the memory of the linear HOCA.

HOCA LINEAR SHIFT REGISTER
\end{comment}
%of a Frobenius LCA $F$ over any given configuration $\ccc\in (\Z_m^n)^{\Z}$is given by the sum of $n$ one-dimensional LCA over $\Z_m$ each of them applied on one element of the history / memory HOCA LINEAR SHIFT REGISTER
%\end{remark}
\begin{comment}
\begin{proposition} Gli HOCA sono esattamente gli MCA la cui fps e' una matrice di Frobenius.
\end{proposition}
\end{comment}
\begin{remark}
Actually, in literature a matrix is in Frobenius normal form if either it or its transpose has a form as in~\eqref{FNF}. Since any matrix in Frobenius normal form is conjugated to its transpose, any Frobenius LCA $F$ is topologically conjugated to a LCA $G$ such that the fps associated with $G$ is the transpose of the fps associated with $G$. In other words, up to a homeomorphism, such LCA $G$, linear HOCA, and Frobenius LCA form the same class and, in particular, the action of $G$ on any configuration $\ccc\in (\Z_m^n)^{\Z}$ is such that
\begin{comment}%%PRECEDENTE
\[
\P_{G(\ccc)}(X) = \M(X)^T \P_{\ccc}(X) =
\begin{bmatrix}
\pp_1(X) c^n(X)\\
c^1(X) + \pp_2(X)c^n(X) \\
\vdots\\
c^{n-1}(X) + \pp_n(X)c^n(X)
\end{bmatrix} \enspace,
\]
\end{comment}
\[
\P_{G(\ccc)}(X) = \M(X)^T \P_{\ccc}(X) =
\begin{bmatrix}
\mm_0(X) c^n(X)\\
c^1(X) + \mm_1(X)c^n(X) \\
\vdots\\
c^{n-1}(X) + \mm_{n-1}(X)c^n(X)
\end{bmatrix} \enspace,
\]
where $\M(X)^T$ is the transpose of the matrix $\M(X)$ associated to the LCA $F$ which is topologically conjugated to $G$.
%LCA $F$ that is a  Frobenius LCA according to ... is topologically conjugated to a LCA $G$ such that the fps associated with  $G$ is the transpose of the fps associated to $F$
%
%formal both a matrix and its transpose are considered matrices in Frobenius normal form. Dire rapidamente che le matrici di frobenius sono anche quelle la cui trasposta e' quella definita come sopra. 
%
%Gli MCA di Frobenius colonna sono top coniugati con quelli riga. Per cui la classe dei Frobenius colonna coincide con la classe dei frobenius riga che coincide con gli HOCA. I Frobenius colonna, analizzati per i fatti loro visti agiscono cosi: FORMULA.
\end{remark}
From now on, we will focus on Frobenius LCA, \ie, matrix presentations of linear HOCA.  Indeed, they allow convenient algebraic manipulations that are very useful to study formal properties of linear HOCA. For example, in~\cite{bruyn1991} and~\cite{kari2000}, the authors proved characterizations for injectivity and surjectivity for LCA in terms of the matrix $\M(X)$ associated to them and which turns out to be decidable  by means of the  characterization of injectivity and surjectivity for LCA over $\Z_m$ shown in~\cite{ito83}.
\begin{proposition}[\cite{bruyn1991, kari2000}]
\label{prop:kari}
Let $\left(\left( \Z_m^n \right)^{\Z}, F\right )$ be a LCA over $\Z^n_m$ and let $\M(X)$ be the matrix associated with $F$. Then, $F$ is injective (resp., surjective) if and only if the determinant of $\M(X)$ is the fps associated with a injective (resp., surjective) LCA over $\Z_m$. 
\end{proposition}
We want to stress that, by Remark~\ref{rem:classes}, Definition~\ref{frobeniusLCA}, and Proposition~\ref{prop:homeo}, one can use the characterizations from Proposition~\ref{prop:kari} for deciding injectivity and surjectivity of linear HOCA. Summarizing, the following result holds.
\begin{proposition}
Injectivity and surjectivity are decidable properties for HOCA of memory size $n$ over  $\Z_m$.
\end{proposition}

In this paper we are going to adopt a similar attitude, \ie, we are going to characterise
the dynamical behaviour of linear HOCA by the properties of the matrices in their matrix
presentation. %$\mathscr{P}$ 
\begin{comment}From now on, we will focus on matrix presentations and turn back to \HOCA\ only
when the general results have to be specialised to \HOCA\ case. 
Indeed, the matrix representation allows convenient algebraic manipulations.
For example, in~\cite{bruyn1991}, the authors proved that a \HOCA\ 
is reversible if and only if some matrix built from the local rule is invertible.
%%% c-era a cap[o
In this paper we are going to adopt a similar attitude, we are going to characterize
the dynamical behaviour of \HOCA\ by the properties of the matrices in their matrix
presentation. The next section provides further motivations for the study of \HOCA\
and the following one will introduce all the basic notions to precisely state
the new results.
\end{comment}

%
\section{Dynamical properties}
\label{sec:dyn}
In this paper we are particularly interested to the so-called \emph{sensitivity to the initial conditions} and \emph{equicontinuity}. As dynamical properties, they represent the main features of instable and stable DDS, respectively. The former is the well-known basic component and essence of the chaotic behavior of DDS, while the latter is a strong form of stability.

Let $(\mathcal{X}, \mathcal{F})$ be a DDS. 
The DDS $(\mathcal{X}, \mathcal{F})$ is \emph{sensitive to the initial conditions} (or simply \emph{sensitive}) if there exists $\varepsilon>0$ such that for any $x\in \mathcal{X}$ and any
$\delta>0$ there is an element $y\in \mathcal{X}$ such that
$d(y,x)<\delta$ and $d(\mathcal{F}^n(y),\mathcal{F}^n(x))>\varepsilon$
for some $n\in\n$. Recall that, by Knudsen's Lemma~\cite{knudsen94}, $(\mathcal{X}, \mathcal{F})$ is sensitive iff $(\mathcal{Y}, \mathcal{F})$ is sensitive where $\mathcal{Y}$ is any dense subset of $\mathcal{X}$ which is $\mathcal{F}$-invariant, \ie,  $\mathcal{F}(\mathcal{Y})\subseteq\mathcal{Y}$.%In~\cite{Ku97}, K\r{u}rka 
%proved that the following conditions are equivalent: $F$ is not sensitive; $F$ is almost equicontinuous;  $F$ admits a $r$-blocking word.\\

In the sequel, we will see that in the context of LCA an alternative way to study sensitivity is via equicontinuity points. An element $x\in \mathcal{X}$ is an \emph{equicontinuity point} for $(\mathcal{X}, \mathcal{F})$ if $\forall\varepsilon>0$
there exists $\delta>0$ such that for all $y\in \mathcal{X}$, $d(x,y)<\delta$ implies that $d(\mathcal{F}^n(y),\mathcal{F}^n(x))<\varepsilon$ for all $n\in\n$. The system $(\mathcal{X}, \mathcal{F})$ is said to be \emph{equicontinuous} if $\forall\varepsilon>0$ there exists $\delta>0$ such that for all $x,y\in \mathcal{X}$,
$d(x,y)<\delta$ implies that $\forall n\in\n,\;d(\mathcal{F}^n(x),\mathcal{F}^n(y))<\varepsilon$.  
%, while
%it is said to be \emph{almost equicontinuous} if the set $E$ of its equicontinuity points is residual (\ie, $E$ contains a countable intersection of dense open subsets). 
Recall that any CA $(\az, F)$ is equicontinuous if and only if there exist two integers $q\in\n$ and  $p>0$ such that $F^q=F^{q+p}$~\cite{Ku97}. Moreover, for the subclass of LCA defined by $n=1$ the following result holds:
\begin{theorem}[\cite{ManziniM99}]
Let %$\mathcal{L}=\structure{\Z_m, r, f}$ 
$(\Z_m^{\Z}, F)$ be a LCA where the local rule $f\colon(\Z_m)^{2r+1}\to\Z_m$ is defined by $2r+1$ coeffiecients  
$m_{-r},\ldots, m_0,\ldots, m_r\in\Z_m$. Denote by $\mathcal{P}$ the set of prime factors of $m$. The following statements are equivalent:
\begin{enumerate}
\item $F$ is sensitive to the initial conditions;
\item $F$ is not equicontinuous;
\item  there exists a prime number $p\in\mathcal{P}$ which does not divide $\gcd(m_{-r},\ldots, m_{-1},m_{1},\ldots, m_r)$.
\end{enumerate}
\end{theorem}  
The dichotomy between sensitivity and equicontinuity still holds for general LCA.
\begin{proposition}\label{prop:dich}
Let $\mathcal{L}=\structure{\Z_m^n, r, f}$ be a LCA where the local rule $f\colon(\Z_m^n)^{2r+1}\to\Z_m^n$ is defined by $2r+1$ matrices  
$\M_{-r},\ldots, \M_0,\ldots, \M_r\in\mat{n}{\Z_m}$. The following statements are equivalent:
\begin{enumerate}
\item $F$ is sensitive to the initial conditions;
\item $F$ is not equicontinuous;
%\item \lim\sup_{t\to\infty} .
\item $\left | \{\M(X)^i, i\geq 1 \} \right | = \infty$.
\end{enumerate}
\end{proposition}
\begin{proof}
The equivalence between $1.$ and $2.$ is a consequence of linearity of $F$ and the Knudsen's Lemma applied on the subset of the finite configurations, \ie, those having a state different from the null vector only in a finite number of cells. Finally, if condition $3.$ is false, then $F$ is equicontinuous, that is $2.$ is false too.
\qed
\end{proof}
An immediate  consequence of Proposition~\ref{prop:dich} is that any characterization of sensitivity to the initial conditions in terms of the matrices defining LCA over $\Z_m^n$ would also provide a characterization of equicontinuity. In the sequel, we are going to show that such a characterization actually exists. First of all, we recall a result that helped in the investigation of dynamical properties in the case $n=1$ and we now state it in a more general form %to be useful 
for LCA over $\Z_m^n$ (immediate generalisation of the result in~\cite{CattaneoDM04,DamicoMM03}). 

Let $\left ( (\Z_m^n)^{\Z}, F \right)$ be a LCA and let $q$ be any factor of $m$. 
%For any configuration $\ccc\in (\Z_m^n)^{\Z}$, we will denote by $\modulo{\ccc}{q}$  the configuration in $(\Z_q^n)^{\Z}$ defined as 
We will denote by $\modulo{F}{q}$ the map $\modulo{F}{q}: (\Z_q^n)^{\Z}\to (\Z_q^n)^{\Z}$ defined as $\modulo{F}{q}(\ccc)=\modulo{F(\ccc)}{q}$, for any $\ccc\in(\Z_q^n)^{\Z}$.
\begin{lemma}[\cite{CattaneoDM04,DamicoMM03}]
\label{lem:dec}
Consider any LCA  $\left ( (\Z_m^n)^{\Z}, F \right)$  with $m=pq$ and $\gcd(p,q)=1$. It holds that the given LCA is topologically conjugated to $\left ( (\Z_p^n)^{\Z} \times (\Z_q^n)^{\Z}, \modulo{F}{p}\times\modulo{F}{q} \right)$.
\end{lemma}
\begin{comment}
\begin{proof}

\end{proof}
\end{comment}
As a consequence of Lemma~\ref{lem:dec}, if $m=p_1^{k_1} \cdots p_l^{k_l}$ is the prime factor decomposition of $m$, any LCA over $\Z_m^n$  is topologically conjugated to the product of LCAs over $\Z^n_{p_{i}^{k_i}}$. %So all the properties which are preserved under product and under topological conjugacy are lifted from LCA over $\Z_{p^{k} }^n$ to LCA over $\Z_{m}^n$. Since sensitivity and equicontinuity are among these properties,
Since sensitivity is preserved under topological conjugacy for DDS over a compact space and the product of two DDS is sensitive if and only if at least one of them is sensitive, we will study sensitivity for Frobenius LCA over $\Z_{p^{k}}^n$. We will show a decidable characterization of sensitivity to the initial conditions for Frobenius LCA over $\Z_{p^{k}}^n$ (Lemma~\ref{frobeniusSensitivity}). Such a decidable characterization together with the previous remarks about the decomposition of $m$, the topological conjugacy involving  any LCA over $\Z_m^n$  and the product of LCAs over $\Z^n_{p_{i}^{k_i}}$, and how sensitivity behaves with respect to a topological conjugacy and the product of DDS, immediately lead to state the main result of the paper.
\begin{theorem}\label{main}
Sensitivity and Equicontinuity are decidable for Frobenius LCA over $\Z_m^n$, or, equivalently, for linear HOCA over $\Z_m$ of memory size $n$.
\end{theorem}

\section{Sensitivity of Frobenius LCA over $\Z^n_{p^k}$}
%======================================================
%======================================================
%Let $\LP$ be the set of the Laurent polynomials with coefficients belonging to a ??? ring???.
% Let $\LPK$ be the set of Laurent polynomials with coefficients belonging to $\Z_{p^k}$.
% MI SEMBRA CHE L'INSIEME PIU' GROSSO $\LP$ SERVA SOLO PER DIMOSTRARE (vedi Teorema molto piu' avanti) LA RICORRENZA  T(m)= un polinomio x T(m-1)+ un polinomio x T(m-2)... ecc CHE NON VALE SOLO SE I COEFFICIENTI SONO SU ZPK
 
%\renewcommand{\LP}{\ensuremath{\Z_{p^k}\left[X,X^{-1}\right]}\xspace}
 
%Denote \LP the set of Laurent polynomials with coefficients over $\Z_{p^k}$. We are
%interested in the degrees of such polynomials both \wrt $X$ and $X^{-1}$.
 
%\subsection{Definizioni e Lemmi di servizio}
In order to study sensitivity of Frobenius LCA over $\Z^n_{p^k}$,  we introduce two concepts about Laurent polynomials. 
\begin{definition}[$deg^+$ and $deg^-$]
  \label{degree}
Given any polynomial $\pp(X)\in\LPK$, the \emph{positive} (resp., \emph{negative}) \emph{degree of $\pp(X)$}, denoted by  $deg^+[\pp(X)]$ (resp., $deg^-[\pp(X)]$) is the maximum (resp., minimum) degree among those of the monomials having both positive (resp., negative) degree and coefficient which is not multiple of $p$. If there is no monomial satisfying both the required conditions, then $deg^+[\pp(X)]=0$ (resp., $deg^-[\pp(X)]$=0).
\end{definition}

\begin{example}
Consider the Laurent polynomial $\pp(X)=4X^{-4}+3X^{-3}+3+7X^{2}+6X^{5}$ with coefficients in  $\Z_8$. Then, $deg^+[\pp(X)]=2$ and $deg^-[\pp(X)]=-3$.
\end{example}

\begin{definition}[Sensitive polynomial]\label{polsens}
A polynomial $\pp(X)\in\LPK$ is \emph{sensitive} if either $deg^+[\pp(X)]>0$ or $deg^-[\pp(X)]<0$. 
As a consequence, a Laurent polynomial $\pp(X)$ is not sensitive iff $deg^+[\pp(X)]=deg^-[\pp(X)]=0$.
\end{definition}
Trivially, it is decidable to decide whether a Laurent polynomial is sensitive.
\begin{remark}
Consider a matrix $\M(X)\in\mat{n}{\LPK}$ in Frobenius normal form. The characteristic polynomial of $\M(X)$ is then
%%%
%old
%%%
%\[
%P_{\M}(y)= y^n-c_0(X)-c_1(X) y -\cdots - c_{n-1}(X) y^{n-1}
%\]
%\[
%P_{\M}(y)= 
%PRECEDENTE
%$(-1)^n\left(-\pp_0(X)-\pp_1(X) y -\cdots - \pp_{n-1}(X) y^{n-1} + y^n\right)$\enspace. \\
%$(-1)^n\left(-\pp_1(X)-\pp_2(X) y -\cdots - \pp_{n}(X) y^{n-1} + y^n\right)$\enspace. 
%\]
$\mathscr{P}(y)=(-1)^n\left(-\mm_0(X)-\mm_1(X) y -\cdots - \mm_{n-1}(X) y^{n-1} + y^n\right)$. %\enspace. 
By the Cayley-Hamilton Theorem, one obtains
\begin{eqnarray}\label{cayleyeq}
%OLD \M(X)^n = c_{n-1}(X) M(X)^{n-1} +\dots + c_{1}(X) M(X)^{1} + c_{0}(X) I 
%PRECEDENTI \M^n(X) = \pp_{n-1}(X) \M(X)^{n-1} +\dots + \pp_{1}(X) \M(X)^{1} + \pp_{0}(X) I \enspace.\\
%\M^n(X) = \pp_{n}(X) \M(X)^{n-1} +\dots + \pp_{2}(X) \M(X)^{1} + \pp_{1}(X) I \enspace.
\M^n(X) = \mm_{n-1}(X) \M(X)^{n-1} +\dots + \mm_{1}(X) \M(X)^{1} + \mm_{0}(X) I \enspace.
\end{eqnarray}
\end{remark}
%%
\begin{comment}
The matrix $\M(X)$ of a Frobenius LCA %concentrates %way 
contains too much information for our purposes. %which is not needed 
Hence, we introduce two new matrices that will focus only on the essential information.
\end{comment}
%The matrix $\M(X)$ of a Frobenius LCA %concentrates %way 
%contains too much information for our purposes. %which is not needed 
We now introduce two further matrices that will allows us to access the information hidden inside $\M(X)$.
\begin{definition}[$\U(X)$, $\LL(X)$, $d^+$, and $d^-$]%{\blue ($\HH_M(X)$ e $\LL_M(X)$)}%{\blue ($\M^{high}(X)$ e $\M^{low}(X)$)}
\label{updown}
For any matrix $\M(X)\in \mat{n}{\LPK}$ in Frobenius normal form %and having  the $n$-th column made of Laurent polynomial $c_0(X), \ldots, c_{n-1}(X)$ with coefficients in $\Z_{p^k}$, 
the matrices %$\M^{high}(X), \M^{low}(X)\in \mat{n}{\LPK}$ 
$\U(X), \LL(X)\in \mat{n}{\LPK}$
associated with $\M(X)$ are the matrices in Frobenius normal %form having  the $n$-th column made of the Laurent polynomial $c^{high}_0(X), \ldots, c^{high}_{n-1}(X)$ and $c^{low}_0(X), \ldots, c^{low}_{n-1}(X)$, 
where each component $\uu_i(X)$ and $\ll_i(X)$ (with $i=0, \ldots, n-1$) of the $n$-th row %$n$-th column %of the vectors 
$\myvec{\uu}(X)$ and $\myvec{\ll}(X)$ of $\U(X)$ and $\LL(X)$, respectively, 
%$\pp^{high}_i(X)$ and $\pp^{low}_i(X)$ of the vectors $\myvec{\pp}^{high}(X)$ and $\myvec{\pp}^{low}(X)$ 
is defined as follows:
%
\begin{comment}
%%%
%OLD
\begin{align*}
c_i^{high}(X)&=
\begin{cases} \text{monomial  of degree } deg^+[c_i(X)] \text{ inside } c_i(X) & \text{if } d_i^+=d^+\\
0 & \text{otherwise}
\end{cases}
\\
c_i^{low}(X)&=
\begin{cases} \text{monomial  of degree } deg^-[c_i(X)] \text{ inside } c_i(X) & \text{if } d_i^-=d^-\\
0 & \text{otherwise}
\end{cases}
\enspace, 
\end{align*}
where $d_i^+=\frac{deg^+[c_i]}{n-i}$, $d_i^-=\frac{deg^-[c_i]}{n-i}$, $d^+=\max\{d_i^+\}$, and $d^-=\min\{d_i^-\}$.
\end{comment}
%%%%
%p
\begin{align*}
\uu_i(X)&=
%\pp_i^{high}(X)&=
\begin{cases} \text{monomial  of degree } deg^+[\mm_i(X)] \text{ inside } \mm_i(X) & \text{if } d_i^+=d^+\\
0 & \text{otherwise}
\end{cases}
\\
%\pp_i^{low}(X)&=
\ll_i(X)&=
\begin{cases} \text{monomial  of degree } deg^-[\mm_i(X)] \text{ inside } \mm_i(X) & \text{if } d_i^-=d^-\\
0 & \text{otherwise}
\end{cases}
\enspace, 
\end{align*}
where $d_i^+=\frac{deg^+[\mm_i(X)]}{n-i}$, $d_i^-=\frac{deg^-[\mm_i(X)]}{n-i}$, $d^+=\max\{d_i^+\}$, and $d^-=\min\{d_i^-\}$.
\end{definition}
\begin{example}\label{exupdown} Consider the following matrix $\M(X)\in\mat{4}{\Z_{49}[X,X^{-1}]}$ in Frobenius normal form
\begin{comment}
%PRECEDENTE
\[
\M(X) =\begin{bmatrix}
 0& 0 &0     & X^{-2}+1+X+2X^8 +14X^{123}   \\
 1& 0 &0     & 3 X^{-3}+3+X^2    \\
 0& 1 &0     & 21 X^{-70}+4 X^{-1}+3 X^4    \\
 0& 0 &1     &7 X^{-35}+ X^{-1}+3    \\
 \end{bmatrix} 
\]
\end{comment}
\[
\M(X) =\begin{bmatrix}
 0& 1 &0     & 0   \\
 0& 0 & 1     & 0    \\
 0& 0 &0     & 1    \\
 X^{-2}+1+X+2X^8 +14X^{123} & 3 X^{-3}+3+X^2 & 21X^{-70}+4 X^{-1}+3 X^4     &7 X^{-35}+ X^{-1}+3    \\
 \end{bmatrix} 
\]
We get $d_0^+=2,\ d_1^+=\frac{2}{3},\ d_2^+=2,\ d_3^+=0$ and $d_0^-=-\frac{1}{2},\ d_1^-=-1,\ d_2^-=-\frac{1}{2},\ d_3^-=-1$. Since $d^+=2$ and $d^-=-1$, it holds that $\uu_0(X)=2 X^8,\ \uu_1(X)=0,\ \uu_2(X)=3 X^4,\ \uu_3(X)=0$ and
$\ll_0(X)=0,\ \ll_1(X)=3 X^{-3},\ \ll_2(X)=0,\ \ll_3(X)=X^{-1}$.
%
%$\pp_0^{high}(X)=2 X^8,\ \pp_1^{high}(X)=0,\ \pp_2^{high}(X)=3 X^4,\ \pp_3^{high}(X)=0$ and
%$\pp_0^{low}(X)=0,\ \pp_1^{low}(X)=3 X^{-3},\ \pp_2^{low}(X)=0,\ \pp_3^{low}(X)=X^{-1}$. 
Therefore,
\begin{comment}
%PRECEDENTI
$$\U(X) =\begin{bmatrix}
 0& 0 &0     & 2 X^8    \\
 1& 0 &0     &  0   \\
 0& 1 &0     & 3 X^4    \\
 0& 0 &1     & 0    \\
 \end{bmatrix} 
$$
and
$$\LL(X) =\begin{bmatrix}
 0& 0 &0     & 0    \\
 1& 0 &0     &   3 X^{-3}  \\
 0& 1 &0     & 0    \\
 0& 0 &1     & X^{-1}    \\
 \end{bmatrix} 
$$
\end{comment}
$$\U(X) =\begin{bmatrix}
 0& 1 &0     &0    \\
 0& 0 &1     &  0   \\
 0& 0 &0     & 1    \\
 2 X^8 & 0 & 3 X^4     & 0    \\
 \end{bmatrix} 
$$
and
$$\LL(X) =\begin{bmatrix}
 0& 1 &0     & 0    \\
 0& 0 &1     &   0  \\
 0& 0 &0     & 1    \\
 0& 3 X^{-3} &0     & X^{-1}    \\
 \end{bmatrix} 
$$

\end{example}

\begin{definition}[$\widehat{\M}(X)$ and $\overline{\M}(X)$]
  \label{partevera}
For any Laurent polynomial $\pp(X)\in\LPK$,  %with coefficients in  $\Z_{p^k}$, 
$\widehat{\pp}(X)$ and $\overline{\pp}(X)$ are defined as the Laurent polynomial obtained from $\pp(X)$ by removing all the monomials having coefficients that are multiple of $p$ and $\overline{\pp}(X)= \pp(X)-\widehat{\pp}(X)$, respectively. These definitions extend component-wise to vectors.
For any matrix $\M(X)\in \mat{n}{\LPK}$ in Frobenius normal form, %and having  the $n$-th column made of Laurent polynomial $c_0(X), \ldots, c_{n-1}(X)$ with coefficients in $\Z_{p^k}$, 
$\widehat{\M}(X)$ and $\overline{\M}(X)$ are defined as the matrix obtained from $\M(X)$ by replacing its %$n$-th column $\myvec{\pp}(X)$ with $\hat{\myvec{\pp}}(X)$ 
$n$-th row $\myvec{\mm}(X)$ with $\widehat{\myvec{\mm}}(X)$
and $\overline{\M}(X)=\M(X)-\widehat{\M}(X)$, respectively.
\end{definition}
It is clear that any matrix $\M(X)\in \mat{n}{\LPK}$ in Frobenius normal form can be written as $\M(X)=\widehat{\M}(X) + p \overline{\M'}(X)$, for some $\M'(X)\in \mat{n}{\LPK}$. 
\begin{comment}
Ad esempio se
$p(x)=4x^{-4}+3x^{-3}+3+7x^{2}+6x^{5}$ su $\Z_8$ abbiamo  $\widehat{p(x)}= 3x^{-3}+3+7x^{2} $ e $\overline{p(x)}=4x^{-4} + 6x^{5}$.
\end{comment}

\begin{definition}[Graph $G_{\M}$]
  \label{CH}
Let $\M(X)\in \mat{n}{\LPK}$ be any matrix in Frobenius normal form.  
 The graph $G_{\M}=\langle V_{\M}, E_{\M}\rangle$ associated with $\M(X)$ is such that $V_{\M}=\{1, \ldots, n\}$ and %$E_{\M}=\{(h,k)| \, \M(X)_{(h,k)}\neq 0\}$. 
 $E_{\M}=\{(h,k)\in V^2_{\M} | \, \M(X)^{h}_k \neq 0\}$.
 Moreover, each edge $(h,k)\in E_{\M}$ is labelled with $\M(X)^{h}_k$. %$\M(X)_{(h,k)}$.
 \end{definition}
Clearly, for any matrix $\M(X)\in \mat{n}{\LPK}$ in Frobenius normal form, any natural $t>0$, and any pair $(h,k)$ of entries, the element $\M^t(X)^{h}_k$ is the sum of the weights of all paths of length $t$ starting from $h$ and ending to $k$, where the weight of a path is the product of the labels of its edges.
%%%%%
\begin{example}\label{ex:mat4x4}
Consider any matrix $\M(X)\in\mat{4}{\LPK}$ in Frobenius normal form. The graph $G_{\M}$ associated with $\M(X)$ is represented %can be associated with the graph 
in Figure~\ref{fig:grafo4x4}. It will help to compute $\M^t(X)^{h}_k$.  %$F^i(k,h)$. E' DEFINITO F? SERVE?
Indeed, $\M^t(X)^{h}_k$ is the sum of the labels of all paths of length $t$ from vertex $k$ to $h$ (labels along edges of the same path multiply).
For example for $(h,k)=(4,4)$ one finds
\begin{align*}
\M^1(X)^{4}_4%F^1(4,4)
&=\mm_3(X) \\
\M^2(X)^{4}_4
%F^2(4,4)
&=(\mm_3(X))^2+\mm_2(X)\\
%F^3(4,4)
\M^3(X)^{4}_4
&=\mm_1(X) + 2\mm_2(X)\mm_3(X) + (\mm_3(X))^3\\
\end{align*}
and so on.
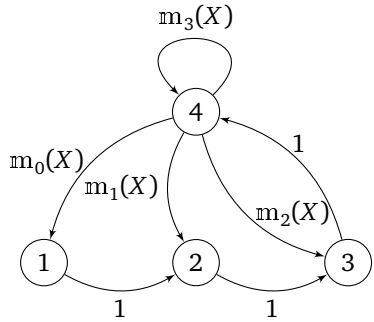
\begin{figure}[htb]
\begin{center}
\begin{tikzpicture}
\tikzset{vertex/.style = {shape=circle,draw,minimum size=1.5em}}
\tikzset{edge/.style = {->,> = latex'}}
% vertices
\node[vertex] (c0) at  (0,0) {$1$};
\node[vertex] (c1) at (2,0) {$2$};
\node[vertex] (c2) at (4,0) {$3$};
\node[vertex] (c3) at  (2,2) {$4$};
%\node [shape=circle,minimum size=1.5em] (a3) at (4.5,0) {q};
%\node [vertex] (cn1) at (6.5,0) {$c_{n-1}$};
%edges
\draw[edge] (c3)  to[bend right]  node[left] {$\mm_0(X)$} (c0)  ;
\draw[edge] (c0) to[bend right] node[below] {1} (c1);
\draw[edge] (c3) to[bend right] node[left] {$\mm_1(X)$} (c1);
\draw[edge] (c1) to[bend right] node[below] {1} (c2);
\draw[edge] (c3) to[bend right] node[right] {$\mm_2(X)$} (c2);
\draw[edge] (c2) to[bend right] node[above] {1} (c3);
\draw[edge] (c3) to[loop]  node[above]  {$\mm_3(X)$} (c3); 
%\path (c2) to node {\dots} (c3) ;
\end{tikzpicture}
\end{center}
\caption{Graph $G_{\M}$ associated with $\M(X)$ from Example~\ref{ex:mat4x4}.}
\label{fig:grafo4x4}
\end{figure}
\end{example}
%%%%
  %\emph{p} $p$ \textit{p} 
%$\mathbbm{p}$ 
%$\mathbbmss{p}$ %$\mathbbmtt{p}$ 
%
%
\begin{lemma}\label{piup}
Let $p>1$ be a prime number and $a, b\geq 0$, $k>0$ be integers such that $1\leq a < p^k$ and $\gcd(a,p)= 1$. Then,
\begin{equation}\label{lemmapiup}
\modulo{a+p b}{p^k}\neq 0
\end{equation}
\end{lemma}
\begin{proof}
For the sake of argument, assume that  $\modulo{a+p  b}{p^k} = 0$. Thus, $a+p  b= c  p^k$ for some $c\geq 0$ and then $a= c  p^k - p b = p(c  p^{k-1}-b )$, that contradicts $\gcd(a,p)= 1$.\qed
\end{proof}

\begin{lemma}\label{fiboext}
Let $p>1$ be a prime number and $h,k$ be two positive integers. Let $l_1,\dots,l_h$  and $\alpha_1,\dots,\alpha_h$ be positive integers such that $l_1<l_2<\cdots <l_h$ and for each $i=1, \ldots, h$ both  $1\leq \alpha_i< p^k$  and $\gcd(\alpha_i,p)=1$ hold. Consider the sequence $b:\Z\to\Z_{p^k}$ defined for any $l\in\Z$ as
\begin{eqnarray}\label{eq1}
b_l&=&\modulo{\alpha_1 b_{l-l_1}+\cdots + \alpha_h b_{l-l_h}}{p^k} \quad \text{ if } l>0\\
\nonumber
b_0&=&1\\
\nonumber
b_l&=&0\quad\text{ if } l<0
\end{eqnarray}
Then, it holds that $\modulo{b_l}{p} \neq 0$ for infinitely many $l\in\N$.
\end{lemma}
\begin{proof}
Set $dz=\left |\{ l\in\N : \modulo{b_l}{p} \neq 0 \}\right |$. For the sake of argument, assume that $dz$ is finite. Then, there exists $h_0\in\N$ such that $\modulo{b_{h_0}}{p} \neq 0$ and $\modulo{b_{l}}{p} = 0$ for all $l>h_0$. Thus, there exist non negative integers $s_0, s_1, \ldots, s_{h-1}$ such that  $b_{l_h+h_0}=ps_0$ and $b_{l_h+h_0-l_i}=ps_i$ for each $i=1, \ldots, h-1$. Equation~\eqref{eq1}  can be rewritten with $l=l_h+h_0$ as 
\begin{equation*}\label{eq2}
b_{l_h+h_0}=\modulo{\alpha_1b_{l_h+h_0-l_1}+\cdots + \alpha_{h-1} b_{l_h +h_0 - l_{h-1}} + \alpha_h b_{h_0}}{p^k}\enspace,
\end{equation*}
which gives 
\begin{equation*}\label{eq2}
ps_0=\modulo{\alpha_1ps_1+\cdots + \alpha_{h-1} ps_{h-1} + \alpha_h b_{h_0}}{p^k}\enspace.
\end{equation*}
Thus, there  must exist an integer $s\geq 0$ such that
%%%
\begin{equation*}\label{eq2}
%\alpha_1ps_1+\cdots + \alpha_{h-1} ps_{h-1} + \alpha_h b_{h_0}=
 p\sum_{i=1}^{h-1} \alpha_i s_i  + \alpha_h b_{h_0} =s p^k + ps_0\enspace,
\end{equation*}
or, equivalently,
\begin{equation*}\label{eq2}
%\alpha_1ps_1+\cdots + \alpha_{h-1} ps_{h-1} + \alpha_h b_{h_0}=
  \alpha_h b_{h_0} =p\left ( s p^{k-1} + s_0 - \sum_{i=1}^{h-1} \alpha_i s_i \right ) \enspace,
\end{equation*}
with $ps_0 < p^k$. If $s p^{k-1} + s_0 - \sum_{i=1}^{h-1} \alpha_i s_i =0$, we get  $\alpha_h b_{h_0}=0$ that contradicts either the assumption $\modulo{b_{h_0}}{p} \neq 0$ or the hypothesis $1\leq \alpha_h< p^k$. Otherwise, $p$ must divide either $\alpha_h$ or $b_{h_0}$. However, that is impossible since $\gcd(\alpha_h,p)=1$ and $\gcd(b_{h_0},p)=1$.
\begin{comment}
\begin{equation*}\label{eq3}
p s=\modulo{\alpha_h b_{h_0}}{p^k}
\end{equation*}
As a consequence, $p$ divides either $\alpha_h$ or $b_{h_0}$. However, that is impossible since $\gcd(\alpha_h,p)=1$ and $\gcd(b_{h_0},p)=1$.
\end{comment}
%%%
\qed 
\end{proof}

For any matrix $\M(X)\in\mat{n}{\LPK}$ in Frobenius normal form, we are now going to study the behavior of $\U^{t}(X)$ and $\LL^{t}(X)$, and, in particular, of their elements $\U^{t}(X)^n_n$ and $\LL^{t}(X)^n_n$. These will turn out to be crucial in order to establish the sensitivity of the LCA defined by $\M(X)$.
To make our arguments clearer we prefer to start with an example.

\begin{ex}\label{monex}
Consider the following matrix $\M(X)\in\mat{4}{\Z_{49}[X,X^{-1}]}$ in Frobenius normal form
\begin{comment}
%PRECDENTE COLONNA
\[
\M(X) =\begin{bmatrix}
 0& 0 &0    & X^{-2}+1+X+16X^6    \\
 1& 0 &0    & 13 X^{-3}+3+X^2    \\
 0& 1 &0    & 34 X^{-1}+8 X^3    \\
 0& 0 &1    & X^{-1}+31    \\
 \end{bmatrix} 
 \in\mat{4}{\Z_{49}[X,X^{-1}]}
\]
\end{comment}
\[
\M(X) =\begin{bmatrix}
 0& 1 &0    &    0 \\
 0& 0 &1   & 0    \\
 0& 0 &0    &  1    \\
 X^{-2}+1+X+16X^6& 13 X^{-3}+3+X^2 &34 X^{-1}+8 X^3    & X^{-1}+31    \\
 \end{bmatrix} 
 \in\mat{4}{\Z_{49}[X,X^{-1}]}
\]
%From the definitions 
One finds $d^+=3/2$ and then
\begin{comment}
%precedente colonna
\[
\U(X) =\begin{bmatrix}
 0& 0 &0     & 16 X^6    \\
 1& 0 &0     &  0   \\
 0& 1 &0     & 8 X^3    \\
 0& 0 &1     & 0    \\
 \end{bmatrix} 
\]
\end{comment}
\[
\U(X) =\begin{bmatrix}
 0& 1 &0     & 0    \\
 0& 0 &1     &  0   \\
 0& 0 &0     & 1    \\
 16 X^6& 0 & 8 X^3 & 0    \\
 \end{bmatrix} 
\]
The graph $G_{\U}$ is represented in Figure~\ref{fig:monex}, along with the values $\U^t(X)^4_4$, $\modulo{\U^t(X)^4_4}{49}$, and $td^+$, for $t=1, \ldots,  8$.
%
\begin{comment}
\begin{figure}[htb]
\begin{center}
\begin{tikzpicture}
\tikzset{vertex/.style = {shape=circle,draw,minimum size=1.5em}}
\tikzset{edge/.style = {->,> = latex'}}
% vertices
\node[vertex] (c0) at  (0,0) {$1$};
\node[vertex] (c1) at (2,0) {$2$};
\node[vertex] (c2) at (4,0) {$3$};
\node[vertex] (c3) at  (2,2) {$4$};
\draw[edge] (c0)  to[bend left]  node[left] {$16X^6$} (c3)  ;
\draw[edge] (c1) to[bend left] node[below] {1} (c0);
\draw[edge] (c2) to[bend left] node[below] {1} (c1);
\draw[edge] (c2) to[bend left] node[near end,left] {$8X^3$} (c3);
\draw[edge] (c3) to[bend left] node[above] {1} (c2);
\end{tikzpicture}
\end{center}
\caption{The graph $G_{\U}$ from Example~\ref{monex}.}
\label{fig:monex}
\end{figure}
\end{comment}
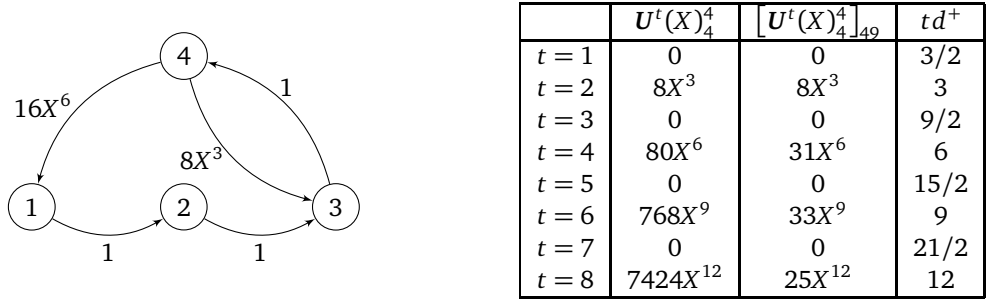
\begin{figure}[htb]
  \centering
  \begin{minipage}{0.45\textwidth}
    \centering
      \begin{tikzpicture}
        \tikzset{vertex/.style = {shape=circle,draw,minimum size=1.5em}}
        \tikzset{edge/.style = {->,> = latex'}}
        % vertices
        \node[vertex] (c0) at  (0,0) {$1$};
        \node[vertex] (c1) at (2,0) {$2$};
        \node[vertex] (c2) at (4,0) {$3$};
        \node[vertex] (c3) at  (2,2) {$4$};
        \draw[edge] (c3)  to[bend right]  node[left] {$16X^6$} (c0)  ;
        \draw[edge] (c0) to[bend right] node[below] {1} (c1);
        \draw[edge] (c1) to[bend right] node[below] {1} (c2);
        \draw[edge] (c3) to[bend right] node[left] {$8X^3$} (c2);
        \draw[edge] (c2) to[bend right] node[above] {1} (c3);
      \end{tikzpicture}
    \end{minipage}
    \begin{minipage}{0.45\textwidth}
      \centering
      \begin{tabular}{|c|c|c|c|}
        \hline
        & $\U^t(X)^4_4$ & $\modulo{\U^t(X)^4_4}{49}$&$td^+$ \\
        \hline
        $t=1$ & $0$ & $0$ &$3/2$\\
        $t=2$ & $8X^3$ & $8X^3$&$3$  \\
        $t=3$ & $0$ & $0$ &$9/2$\\
        $t=4$ & $80X^6$  & $31X^6$ &$6$\\
        $t=5$ & $0$ & $0$ &$15/2$\\
        $t=6$ & $768X^9$ & $33X^9$ &$9$\\
        $t=7$ & $0$ & $0$&$21/2$ \\
        $t=8$ & $7424 X^{12}$ & $25X^{12}$&$12$ \\
        \hline
      \end{tabular}
    \end{minipage}
\caption{The graph $G_{\U}$ (on the left), and the  values $\U^t(X)^4_4$, $\modulo{\U^t(X)^4_4}{49}$, and $td^+$, for $t=1, \ldots,  8$ (on the right) from Example~\ref{monex}.}
\label{fig:monex}
\end{figure}
\end{ex}
\begin{notation}
For a sake of simplicity, for any given matrix $\M(X)\in\mat{n}{\LPK}$ in Frobenius normal form, from now on we will denote by  $\uu^{(t)}(X)$ and $\ll^{(t)}(X)$ the elements $(\U^t(X))_{n}^n$ and $\LL^{t}(X)^n_n$, respectively. %$\M^{high}(X)_{(n,n)}^m$
\end{notation}
\begin{lemma}\label{sempremonomio}
Let $\M(X)\in \mat{n}{\LPK}$ be a matrix such that $\M(X)=\widehat{\NN}(X)$ for some matrix $\NN(X) \in \mat{n}{\LPK}$ in Frobenius normal form. 
For any natural $t>0$, $\uu^{(t)}(X)$ (resp., $\ll^{(t)}(X)$) is either null or a monomial of degree $td^+$ (resp., $td^-$). %or $(\U^t(X))_{n}^n=0$.  
\begin{comment}
Sia $F$ una matrice $n\times n$ di Frobenius su $\Z_{p^k}$ depurata da monomi con coefficienti multipli di $p$.
Assumiamo che $d^+>0$. Allora $(F^{high})^m(n,n)$ \`e un monomio di grado $md^+$ oppure \`e uguale a $zero$
\end{comment}
\end{lemma}
\begin{proof} 
We show that the statement is true for $\uu^{(t)}(X)$ (the proof concerning $\ll^{(t)}(X)$ is identical by replacing $d^+$, $\U(X)$ and related elements with $d^-$, $\LL(X)$ and related elements). For  each $i\in V_{\U}$, let $\gamma_i$ be the simple cycle of $G_{\U}$ from $n$ to $n$ and passing  through the edge $(n,i)$. Clearly, $\gamma_i$ is the path $n\to i\to i+1\ldots\to n-1\to n$ (with $\gamma_n$ the self-loop $n\to n$) of length $n-i+1$ and its weight is the monomial $\uu_{i-1}(X)$ of degree $(n-i+1)d^+$. We know that $\uu^{(t)}(X)$ is the sum of the weights of all cycles of length $t$ starting from $n$ and ending to $n$ in $G_{\U}$ if at least one of such cycles exists, $0$, otherwise. In the former case, each of these cycles can be  decomposed in a certain number $s\geq 1$ of simple cycles $\gamma_{j_1}^1,\ldots, \gamma_{j_s}^s$ of lengths giving sum $t$, \ie,  such that $\sum_{i=1}^{s} (n-j_i+1) = t$. Therefore, $(\U^t(X))_{n}^n$ is a monomial of degree $\sum_{i=1}^{s} (n-j_i+1)d^+ = t d^+$.\qed
\end{proof}
\begin{lemma}\label{tdim}
Let $\M(X)\in\mat{n}{\LPK}$
be any matrix in Frobenius normal form.
For every integer $t\geq 1$ both the following recurrences hold %Si noti che $F^0=I$.
%%%%
\begin{comment}
\begin{eqnarray}\label{eqgen1}
\uu^{(t)}(X) &=&  %c_{n-1}(X)V(m-1)+ \cdots + c_{1}(X)V(m-n+1)+  c_0 V(m-n)(X)
\mm_{n-1}(X)\uu^{(t-1)}(X)+ \cdots + \mm_{1}(X)\uu^{(t-n+1)}(X) +  \mm_0(X)\uu^{(t-n)}(X) 
\end{eqnarray}
\end{comment}
%%%%%
\begin{eqnarray}\label{eqgen1}
\uu^{(t)}(X) &=&  %c_{n-1}(X)V(m-1)+ \cdots + c_{1}(X)V(m-n+1)+  c_0 V(m-n)(X)
\uu_{n-1}(X)\uu^{(t-1)}(X)+ \cdots + \uu_{1}(X)\uu^{(t-n+1)}(X)+  \uu_0(X)\uu^{(t-n)}(X)\\
\ll^{(t)}(X) &=&  %c_{n-1}(X)V(m-1)+ \cdots + c_{1}(X)V(m-n+1)+  c_0 V(m-n)(X)
\ll_{n-1}(X)\ll^{(t-1)}(X)+ \cdots + \ll_{1}(X)\ll^{(t-n+1)}(X)+  \ll_0(X)\ll^{(t-n)}(X)
\end{eqnarray}
with the related initial conditions 
\begin{eqnarray}
\uu^{(0)}(X) &=& 1, \quad \ll^{(0)}(X) = 1\label{eqgen2a} \\ 
\uu^{(l)}(X) &=&0, \quad \ll^{(l)}(X) =0\quad \text{for $l<0$.} 
\label{eqgen2b}
\end{eqnarray}
\end{lemma}
\begin{proof}
We show the recurrence involving $\uu^{(t)}(X)$ (the proof for $\ll^{(t)}(X)$ is identical by replacing $\U(X)$ and its elements with $\LL(X)$ and its elements). Since $\U(X)$ is in Frobenius normal form too, by~\eqref{cayleyeq}, Recurrence~\eqref{eqgen1} holds for every $t\geq n$. It is clear that $\uu^{(0)}(X)=1$. Furthermore, by the structure of the graph $G_{\U}$ and the meaning of $\U(X)^n_{n}$, Equation~\eqref{eqgen1} is true under the initial conditions~\eqref{eqgen2a} and~\eqref{eqgen2b} for each $t=1, \ldots, n-1$.
\qed
\end{proof}

\begin{lemma}\label{infinitim}
Let $\M(X)\in \mat{n}{\LPK}$ be a matrix such that $\M(X)=\widehat{\NN}(X)$ for some matrix $\NN(X) \in \mat{n}{\LPK}$ in Frobenius normal form. Let $\upsilon(t)$ (resp., $\lambda(t)$) be the coefficient of $\uu^{(t)}(X)$ (resp., $\ll^{(t)}(X)$). It holds that $\gcd[\upsilon(t),p]=1$ (resp., $\gcd[\upsilon(t),p]=1$),  for infinitely many $t\in\N$. 
\\
In particular, if the value  $d^+$ (resp.,  $d^-$) associated with $\M(X)$ is non null, then for infinitely many $t\in\N$ both $\modulo{\uu^{(t)}(X)}{p^k}\neq 0$ and $deg(\modulo{\uu^{(t)}(X)}{p^k})\neq 0$ (resp., $\modulo{\ll^{(t)}(X)}{p^k}\neq 0$ and $deg(\modulo{\ll^{(t)}(X)}{p^k})\neq 0$) hold. 
In other terms, if $d^+>0$ (resp.,  $d^-<0$) then $|\{\uu^{(t)}(X), t\geq 1\}|=%|\{\U^t(X), t\geq 1\}|
\infty$ (resp., $|\{\ll^{t}(X), t\geq 1\}|%=|\{\LL^t(X), t\geq 1\}|
=\infty$).%, and in that case we say that $\U^t(X)$ diverges (resp., $\LL^t(X)$ diverges).
\end{lemma}
\begin{proof}
%Let $\upsilon(t)$ be the coefficient of $u(t)$. 
We show the statements concerning $\upsilon(t)$, $\U(X)$, $\uu^{(t)}(X)$, and $d^+$. %Assume that $d^+>0$ and 
Replace $X$ by $1$ in the matrix $\U(X)$. Now, the coefficient $\upsilon(t)$ is just the element of position $(n,n)$ in the $t$-th power of the obtained matrix $\U(1)$. Over $\U(1)$, the thesis of Lemma~\ref{tdim} is still valid replacing $\uu^{(t)}(X)$ by $\upsilon(t)$. Thus, for every $t\in\N$
%%%%%
\begin{comment}
\begin{eqnarray}
\upsilon(t) &=&  c_{n-1}\upsilon(t-1)+ \cdots +  c_1 \upsilon(t-n+1)+c_0 \upsilon(t-n)\\
\end{eqnarray}
\end{comment}
%%%%%
\begin{eqnarray*}
\upsilon(t) &=&  \uu_{n-1}(1)\upsilon(t-1)+ \cdots +  \uu_1(1) \upsilon(t-n+1)+\uu_0(1) \upsilon(t-n)
\end{eqnarray*}
with initial conditions
\begin{eqnarray*}\label{eqgen3}
\upsilon(0)&=& 1 \\
\upsilon(l)&=&0\quad \text{for $l<0$,}
\end{eqnarray*}
where each $\uu_i(1)$ is the coefficient of the monomial $\uu_i(X)$ inside $\U(X)$. Thus, it follows that 
\begin{equation*}\label{eq55}
\modulo{\upsilon(t)}{p^k} = 
\modulo{\uu_{n-1}(1)\upsilon(t-1)+ \cdots +  \uu_1(1) \upsilon(t-n+1)+\uu_0(1) \upsilon(t-n)}{p^k}
\end{equation*}
By Lemma~\ref{fiboext} we obtain that $\gcd[\upsilon(t),p]=1$ (and so $\modulo{\upsilon{(t)}}{p^k}\neq 0$, too) for infinitely many $t\in\N$. In particular, if the value  $d^+$ associated with $\M(X)$ is non null, then, by the structure of $G_{\U}$ and Lemma~\ref{sempremonomio},  both $\modulo{\uu^{(t)}(X)}{p^k}\neq 0$ and $deg(\modulo{\uu^{(t)}(X)}{p^k})\neq 0$ hold
for infinitely many $t\in\N$, too. Therefore, $|\{\uu^{(t)}(X), t\geq 1\}|=%|\{\U^t(X), t\geq 1\}|=
\infty$. The same proof runs for the statements involving $\lambda(t)$, $\LL(X)$, $\uu^{(t)}(X)$, and $d^-$ provided that these replace $\upsilon(t)$, $\U(X)$, $\uu^{(t)}(X)$, and $d^+$, respectively.%, and one assumes $d^-<0$ instead of $d^+>0$. 
\qed
\end{proof}
\begin{comment}
As an immediate consequence of Lemma~\ref{infinitim} one obtains the following.

\begin{teo}\label{union}
Let $\M(X)\in \mat{n}{\LPK}$ be a matrix such that $\M(X)=\widehat{\NN(X)}$ for some matrix $\NN(X) \in \mat{n}{\LPK}$ in Frobenius normal form. If the value  $d^+$ associated with $\M(X)$ is non null, then for infinitely many $t\in\N$ it holds that $deg(\modulo{\uu^{(t)}(X)}{p})\neq 0$ and then
$$deg(\modulo{\uu^{(t)}(X)}{p^k})\neq 0\enspace, $$ ERA MODULO P
or, equivalently, $|\{\U^t(X), t\geq 1\}|=\infty$ (\ie, in other words, $\U^t(X)$ diverges).
\end{teo}
\end{comment}
The following Lemma puts in relation the behavior of $\uu^{(t)}(X)$ or $\ll^{(t)}(X)$ with that of $\widehat{\M}^{\, t}(X)_n^n$, 
%divergence of $\U^t(X)$ or $\LL^t(X)$ with that of $\widehat{\M}^{\, t}(X)$, 
for any matrix $\M(X)\in \mat{n}{\LPK}$  in Frobenius normal form. 
\begin{lemma}\label{corollario}
Let $\M(X)\in \mat{n}{\LPK}$ be a matrix in Frobenius normal form. 
\begin{comment}
If $|\{\U^t(X), t\geq 1\}|=\infty$ (\ie, $\U^t(X)$ diverges) then $|\{\widehat{\M}^{\, t}(X), t\geq 1\}|=\infty$ (\ie, $\widehat{\M}^{\, t}(X)$ also diverges).
\end{comment}
If either $|\{\uu^{(t)}(X), t\geq 1\}|=\infty$ or $|\{\ll^{(t)}(X), t\geq 1\}|=\infty$ %(that implies that $\U^t(X)$ diverges) 
then $|\{\widehat{\M}^{\, t}(X)_n^n, t\geq 1\}|=\infty$.
%(that implies that $|\{\widehat{\M}^{\, t}(X), t\geq 1\}|=\infty$, \ie, $\widehat{\M}^{\, t}(X)$ also diverges).
\begin{comment}
DOMANDA
$|\{\uu^t(X), t\geq 1\}|=\infty$ SE E SOLO SE $|\{\U^t(X), t\geq 1\}|=\infty$??
\\

$|\{\widehat{\M}^{\, t}(X)_n^n, t\geq 1\}|=\infty$ SE E SOLO SE $\widehat{\M}^{\, t}(X)$??
\end{comment}
\end{lemma}
\begin{proof}
Assume that $|\{\uu^{(t)}(X), t\geq 1\}|=\infty$. Since $G_{\U}$ is a subgraph of $G_{\widehat{\M}}$ (with different labels), for each integer $t$ from Lemma~\ref{infinitim} applied to $\widehat{\M}(X)$, the cycles of length $t$ in $G_{\widehat{\M}}$ with weight containing a monomial of degree $td^+$ are exactly the cycles of length $t$ in $G_{\U}$. Therefore, it follows that $|\{\widehat{\M}^{\, t}(X)_n^n, t\geq 1\}|=\infty$. The same argument on $G_{\LL}$ and involving $d^-$ allows to prove the thesis if $|\{\ll^{(t)}(X), t\geq 1\}|=\infty$.
\begin{comment}
\\

Remark that given a matrix 
$\M(X)\in \mat{n}{\LP}$ in Frobenius normal form in which the coefficients of  the Laurent polynomials are not multiple of $p$, the infinite
set of values in Lemma~\ref{infinitim} for which there exists a cycle from $n$ to $n$ in $G_{\M^{high}}$ is a subset of the values
for which there exists a cycle from $n$ to $n$ in $G_{\M}$ since $G_{\M^{high}}$ is a subgraph of $G_{\M}$. Indeed, the following result holds.
\end{comment}
\end{proof}
We are now able to present and prove the main result of this section. It shows a decidable characterization of sensitivity for Frobenius LCA over $\Z_{p^k}^n$.
\begin{lemma}\label{frobeniusSensitivity}
Let $\left ( (\Z_{p^k}^n)^{\Z}, F \right)$ be any Frobenius LCA over $\Z_{p^k}^n$ and let $(\mm_0(X), \ldots, \mm_{n-1})$ be the $n$-th row of the matrix $\M(X)\in\mat{n}{\LPK}$ in Frobenius normal form associated with $F$. Then, $F$ is sensitive to the initial conditions if and only if $\mm_i(X)$ is sensitive for some $i\in[0,n-1]$.
\end{lemma}
%%%%
\begin{proof}
Let us prove the two implications separately.
\\
Assume that all $\mm_i(X)$ are not sensitive. Then, $\widehat{\M}(X)\in \mat{n}{\Z_{p^k}}$, \ie, it does not contain the formal variable $X$, and $\M(X)=\widehat{\M}(X)+p\M'(X)$, for some $\M'(X)\in\mat{n}{\LPK}$ in Frobenius normal form. Therefore, for any integer $t>0$, $\M^t(X)$ is the sum of terms, each of them consisting of a product in which $p^j$ appears as factor, for some natural $j$ depending on $t$ and on the specific term which $p^j$ belongs to. Since every element of $\M^t(X)$ is taken modulo $p^k$, for any natural $t>0$ it holds that in each term of such a sum $p^j$ appears with $j\in[0, k-1]$ (we stress that $j$ may depend on $t$ and on the specific term of the sum, but it is always bounded by $k$). Therefore, $|\{\M^i(X): i>0\}|<\infty$ and so, by Proposition~\ref{prop:dich}, $F$ is not sensitive to the initial conditions.

Conversely, suppose that $\mm_i(X)$ is sensitive for some $i\in[0,n-1]$ and $d^+>0$ (the case $d^-<0$ is identical). By Definition~\ref{partevera}, for any natural $t>0$ there exists a matrix $\M'(X)\in \mat{n}{\LPK}$ such that $\M^t(X)=\widehat{\M}^t(X) + p\M'(X)$. %IL FATTO CHE $M'$ DIPENDE DA $t$ E' UN PROBLEMA? 
By a combination of Lemmata~\ref{infinitim} and~\ref{corollario}, we get $|\{\widehat{\M}^{\, t}(X)_n^n, t\geq 1\}|=\infty$ and so, by  Lemma~\ref{piup}, $|\{\M^t(X)_n^n, t\geq 1\}|=\infty$ too. Therefore, it follows that $|\{\M^t(X), t\geq 1\}|=\infty$ and, by Proposition~\ref{prop:dich}, we conclude that $F$ is sensitive to the initial conditions.
%$\widehat{\M}^t(X)$ diverges and, in particular, its element $\widehat{\M}^t(X)^n_n$. %BISOGNEREBBE PARLARE DELL' ELEMENTO RIGA $n$ COLONNA $n$ SIA QUI CHE NEI LEMMI (NEL PRIMO OK NEL SECONDO NO). 
%Thus, by  Lemma~\ref{piup}, $\M^t(X)_n^n$ diverges too, and, by Proposition~\ref{prop:dich}, we conclude that $F$ is sensitive to the initial conditions.
\qed
\end{proof}

\section{Perspectives}

\label{sec:periodic-non-uniform}
%%%
\begin{comment}
%ENRICO
In this paper we studied equicontinuity and sensitivity to the initial conditions of linear HOCA. The study has interesting perspectives. Indeed,
this class is equivalent to additive \pnuca, a class of non-uniform CA for which little is known. Indeed, in~\cite{DENNUNZIO201473},
injectivity and surjectivity properties are characterized. The following theorem shows that any linear \HOCA\ can be transformed into
an additive \pnuca and vice-versa. Hence the results abot the dynamics of linear \HOCA\ also apply to \pnuca.
\smallskip
\end{comment}
In this paper we have studied equicontinuity and sensitivity to the initial conditions for linear HOCA over $\Z_m$ of memory size $n$, providing decidable characterizations for these properties. 
Such characterizations extend the ones shown in~\cite{ManziniM99} for linear cellular automata (LCA) over $\Z_m^n$ in the case $n=1$ and have an impact in many applications, for instance those concerning the design of secret sharing schemes~\cite{MARTINDELREY20051356,chai2005}, data compression and image processing~\cite{Gu2000}, where linear HOCA are involved and a chaotic behavior is required.  
We also proved that linear HOCA over $\Z_m$ of memory size $n$ form a class that is indistinguishable from a subclass of LCA (namely, the subclass of Frobenius LCA) over $\Z_m^n$. This enables to decide injectivity and surjectivity for linear HOCA over $\Z_m$ of memory size $n$ by means of the decidable characterizations of injectivity and surjectivity provided  in~\cite{bruyn1991} and~\cite{kari2000} for LCA over $\Z^n_m$. A natural and pretty interesting research direction consists of investigating other chaotic properties as transitivity and expansivity for linear HOCA. Possible characterizations of the latter properties would be useful in cryptographic applications since they would allow to design schemes that make attacks even more difficult. Furthermore, all the mentioned dynamical properties, including sensitivity and equicontinuity, could be studied for the whole class of LCA over $\Z_m^n$ (more difficult task).  Besides leading to a generalization of the shown results, such investigations would indeed help to understand also the behavior of other models used in applications, as for instance linear \emph{non-uniform cellular automata} over $\Z_m$. Let us explain precisely the meaning of such a sentence. 

First of all, recall that a \emph{non-uniform cellular automaton} ($\nu$-CA) over the alphabet $S$ is a structure $\structure{S, \{f_j, r_j\}_{j\in\Z}}$ defined by a family of local rules $f_j: S^{2r_j+1}\to S$, each of them of its own radius $r_j$. Similarly to CA, the global rule of a $\nu$-CA is the map $F_{\nu}:\az\to\az$ defined as
\begin{equation*}
\forall c\in \az, \quad \forall i\in\Z,\qquad F_{\nu}(c)_i=f_i(c_{i-r}, \ldots c_{i+r})\enspace,
\end{equation*}
and, so, $(\az, F_{\nu})$ is the DDS associated with a given $\nu$-CA. A $\nu$-CA over the alphabet $S=\Z_m$ is linear if all its local rules are linear.  We are going to focus our attention on an interesting class of (linear) $\nu$-CA, namely, the (linear) periodic ones. A \emph{periodic $\nu$-CA ($\pi\nu$-CA)} is a $\nu$-CA satisfying the following condition: there exists an integer $n>0$ (called the structural period) such that $f_j=f_{j\mod n}$ for any $j\in\Z$. It is clear that a $\pi\nu$-CA of structural period $n$ is defined by the local rules $f_0$, \ldots, $f_{n-1}$, which, without loss of generality, can be all assumed to have a same radius $r$. 
The following result holds.
\begin{proposition}
Every linear $\pi\nu$-CA over $\Z_m$ of structural period $n$ is topologically conjugated to a LCA over $\Z^n_m$, and vice versa. In other terms, up to a homeomorphism, the class of linear $\pi\nu$-CA over $\Z_m$ of structural period $n$ is nothing but the class of LCA over $\Z^n_m$.
\end{proposition}
\begin{proof}
Consider any linear $\pi\nu$-CA of structural period $n$, radius $r$, and local rules $f_0$, \ldots, $f_{n-1}$, where, for each $j=0, \ldots, n-1$, the rule $f_j$ is defined as follows
$$
\forall (x_{-r},\ldots,x_r)\in \Z_m^{2r+1}, \quad f_j(x_{-r},\ldots,x_r) = \sum_{i=-r}^r a_{j,i} x_i\enspace.
$$
Set $s=-\lceil r / n \rceil$ and let $\left( ( \Z_m^n)^{\Z}, F \right )$ be the LCA over $\Z_m^n$ such that its local rule is defined by the matrices ${\M}_{-s}, \ldots,  {\M}_{0}, \ldots, {\M}_{s}\in\mat{n}{\Z_m}$, where, for each $\ell\in [-s,s]$, 
\begin{align*}
  {\M}_{\ell} =
  \begin{pmatrix}
    a_{0,\ell n} & a_{0,\ell n+1} & \cdots & a_{0,\ell n+n-1} \\
    a_{1,\ell n-1} & a_{1, \ell n} & \cdots & a_{1,\ell n+n-2} \\
    \vdots & \vdots & \ddots & \vdots \\
    a_{n-1,\ell n-n+1} & a_{n-1,\ell n-n+2} & \cdots & a_{n-1,\ell n}
  \end{pmatrix}\enspace,
\end{align*}
with~$a_{j,i} = 0$ whenever~$i \notin [-r,r]$. One gets that $\varphi\circ F_{\nu}=F\circ\varphi$, where 
 $\varphi:\Z_m^{\Z}\to \left( \Z_m^n \right)^{\Z}$ is the homeomorphism associating any configuration $c\in\Z_m^{\Z}$ with the element $\varphi(c)\in \left( \Z_m^n \right)^{\Z}$ such that for any $i\in\Z$ and for each $j=0, \ldots, n-1$, the $(j+1)$--th component of the vector $\psi(c)_i\in \Z_m^n$ is $\psi(c)^{j+1}_i=c_{in+j}$. Then, the linear $\pi\nu$-CA $F_{\nu}$ is topologically conjugated to the LCA $F$. Similarly, one proves that every LCA over $\Z_m^\pi$ is topologically conjugated to a linear $\pi\nu$-CA over $\Z_m$ of structural period $\pi$.
 \qed
\begin{comment}
\begin{equation*}
\psi(c)_i= \begin{pmatrix}
    c_{ip} \\
    \vdots \\
    c_{ip+p-1}
  \end{pmatrix}
\end{equation*}
\end{comment}
\end{proof}
This result, the fact that little is known for linear $\nu$-CA, and that linear $\pi\nu$-CA are used in many applications (for instance, as pointed out in~\cite{kari2000}, they can be used as subband encoders for compressing signals and images~\cite{Shapiro93}) further motivate the investigation of LCA over $\Z_m^n$ in the next future. 
\section*{Acknowledgements}

Enrico Formenti acknowledges the partial support from the project PACA APEX FRI.
Alberto Dennunzio and Luca Manzoni were partially supported by Fondo d'Ateneo (FA)  2016 of Università degli Studi di Milano Bicocca: ``Sistemi Complessi e Incerti: teoria ed applicazioni''.
Antonio E. Porreca was partially supported by Fondo d'Ateneo (FA) 2015 of Università degli Studi di Milano Bicocca: ``Complessità computazionale e applicazioni crittografiche di modelli di calcolo bioispirati''.

\bibliographystyle{plain}
\bibliography{\jobname.bib}

\end{document}